\documentclass[journal]{IEEEtran}
\usepackage{graphicx}
\usepackage{subfigure}
\usepackage{mathtools,amssymb}
\usepackage{algpseudocode}
\usepackage{algorithm}
\usepackage{amsmath}
\usepackage{epstopdf}
\usepackage{color}
\usepackage{stfloats}
\usepackage{booktabs} 
\usepackage{underscore}
\usepackage{cases}
\usepackage{cite}
\usepackage{enumerate}
\usepackage[colorlinks, linkcolor=red, anchorcolor=red, citecolor=red]{hyperref}

\floatname{algorithm}{Algorithm}

\begin{document}
	
\title{OPP-Miner: Order-preserving sequential pattern mining}

\author{Youxi Wu \IEEEmembership{Member,~IEEE}, Qian Hu, Yan Li, Lei Guo, Xingquan Zhu, \IEEEmembership{Senior Member,~IEEE,} and Xindong Wu \IEEEmembership{Fellow,~IEEE}% <-this % stops a space	
	
\thanks{*Corresponding author: Y. Li}
\thanks{Y. Wu is with the School of Artificial Intelligence, Hebei University of Technology, Tianjin 300401, China, and also with the Hebei Key Laboratory of Big Data Computing, Tianjin 300401, China.}
\thanks{Q. Hu is with the School of Artificial Intelligence, Hebei University of Technology, Tianjin 300401, China.}
\thanks{Y. Li is with the School of Economics and Management, Hebei University of Technology, Tianjin 300401, China (e-mail:lywuc@163.com).}
\thanks{L. Guo is with the State Key Laboratory of Reliability and Intelligence of Electrical Equipment, Hebei University of Technology, Tianjin 300401, China.}
\thanks{X. Zhu is with the Department of Computer, Electrical Engineering and Computer Science, Florida Atlantic University, FL 33431, USA.}
\thanks{X. Wu is with the Mininglamp Academy of Sciences, Mininglamp Technology, Beijing 100084, China, and also with the Key Laboratory of Knowledge Engineering with Big Data (Hefei University of Technology), Ministry of Education, Hefei 230009, China.}}

\maketitle

\begin{abstract}
A time series is a collection of measurements in chronological order. Discovering patterns from time series is useful in many domains, such as stock analysis, disease detection, and weather forecast. To discover patterns, existing methods often convert time series data into another form, such as nominal/symbolic format, to reduce dimensionality, which inevitably deviates the data values. Moreover, existing methods mainly neglect the order relationships between time series values. To tackle these issues, inspired by order-preserving matching, this paper proposes an Order-Preserving sequential Pattern (OPP) mining method, which represents patterns based on the order relationships of the time series data. An inherent advantage of such representation is that the trend of a time series can be represented by the relative order of the values underneath the time series data. To obtain frequent trends in time series, we propose the OPP-Miner algorithm to mine patterns with the same trend (sub-sequences with the same relative order). OPP-Miner employs the filtration and verification strategies to calculate the support and uses pattern fusion strategy to generate candidate patterns. To compress the result set, we also study finding the maximal OPPs. Experiments validate that OPP-Miner is not only efficient and scalable but can also discover similar sub-sequences in time series. In addition, case studies show that our algorithms have high utility in analyzing the COVID-19 epidemic by identifying critical trends and improve the clustering performance. The algorithms and data can be downloaded from \url{https://github.com/wuc567/Pattern-Mining/tree/master/OPP-Miner}.
\end{abstract}


\IEEEpeerreviewmaketitle

\newtheorem{theorem}{Theorem}
\newtheorem{definition}{Definition}
\newtheorem{example}{Example}
\newenvironment{proof}{\textit{Proof:}}

%order-preserving matching
\section{Introduction}
\IEEEPARstart
	{S}{equential} pattern mining~\cite{Fournier-Viger2017a}, an important branch of data mining, aims to find interesting sub-sequences (also known as patterns) from a given sequence. Analysing these potentially useful patterns is useful for decision making in many domains, such as feature extraction for sequence classification \cite {wu2021tcyb}, text analysis \cite{qiang2020tkde}, disease detection~\cite{Ghosh2017Septic}, event log analysis~\cite{ fournier2020event}, and network clickstream analysis~\cite{Nishimuraa2018latent}. Various methods have been proposed to mine sequential patterns. For example, negative sequential pattern mining~\cite{Dong2020e-RNSP, Dong2019top-k} was proposed to detect fraud patterns. High utility pattern mining~\cite{Truong2020EHAUSM, Gan2020HUOPM, fournier2019, song2021kais} was proposed to avoid mining frequent but unimportant patterns. Gap constraint sequential pattern mining~\cite{Wu2014Mining, Wu2018NOSEP,Truong2019Efficient} was also proposed to make mining patterns more flexible. All these methods have been applied to find valuable information.
	
	Although existing methods have significantly advanced the field, there are still many limitations in processing time series data. More specifically, a major challenge of time series data stems from its high dimensionality and continuity characteristics, which makes it difficult to apply traditional sequential pattern mining methods directly to time series analysis. To tackle this issue, traditional time series representation methods often transform numerical data to another form for dimensionality reduction, such as the segmented representation~\cite{Keogh2001online} and symbolic representation~\cite{Lin2007Experiencing} methods. Such methods not only depend on parameter settings where a variance in parameters may lead to dramatic difference in the results \cite {li2021apind,wu2022ida}, they will also lose some important information during the process and partially break the continuity of the time series. To discover useful information underneath the time series data, it is critical to explore appropriate time series representation and mining methods.
		
	Because time series data consists of numeric values, users are often interested in finding trends in a time series instead of discovering specific character patterns \cite {wu2021tmis,wu2021insweak}. For example, the trend of stock price in several consecutive days is important for investors. Therefore, it is necessary to find a new pattern that can represent the trend of a time series. Recently, order-preserving pattern matching was proposed~\cite{Kim2014Order} to find all sub-sequences with the same relative order in a time series as the given pattern. Although it can locate sub-sequences with the same trend as the given pattern, this approach is rather limited because it cannot find frequent patterns, and also cannot find novel patterns previously unknown. Inspired by order-preserving matching~\cite{Kim2014Order}, we propose order-preserving sequential pattern (OPP) mining, which can find the frequent sub-sequences with the same relative order in a time series. An illustrative example is as follows.
	
\begin{example}
\label{example1}
	Given a time series of product sales for 16 days,  \textbf{s} $ = $ ($\textit{s}_{1} $, $\textit{s}_{2} $, $ \textit{s}_{3} $, $ \textit{s}_{4} $, $ \textit{s}_{5} $, $ \textit{s}_{6} $, $ \textit{s}_{7} $, $ \textit{s}_{8} $, $ \textit{s}_{9} $, $ \textit{s}_{10} $, $ \textit{s}_{11} $, $ \textit{s}_{12} $, $ \textit{s}_{13} $, $ \textit{s}_{14} $, $ \textit{s}_{15} $, $ \textit{s}_{16} $) $ = $ (11, 10, 21, 25, 12, 14, 18, 19, 26, 13, 16, 20, 24, 30, 15, 17) shown in Fig.~\ref{FIG:1,intro}. We aim to mine frequent OPPs with $ \textit{minsup} = $ 3.  
\end{example}

\begin{figure}[ht]
	\centering
	\includegraphics[width=0.45\textwidth]{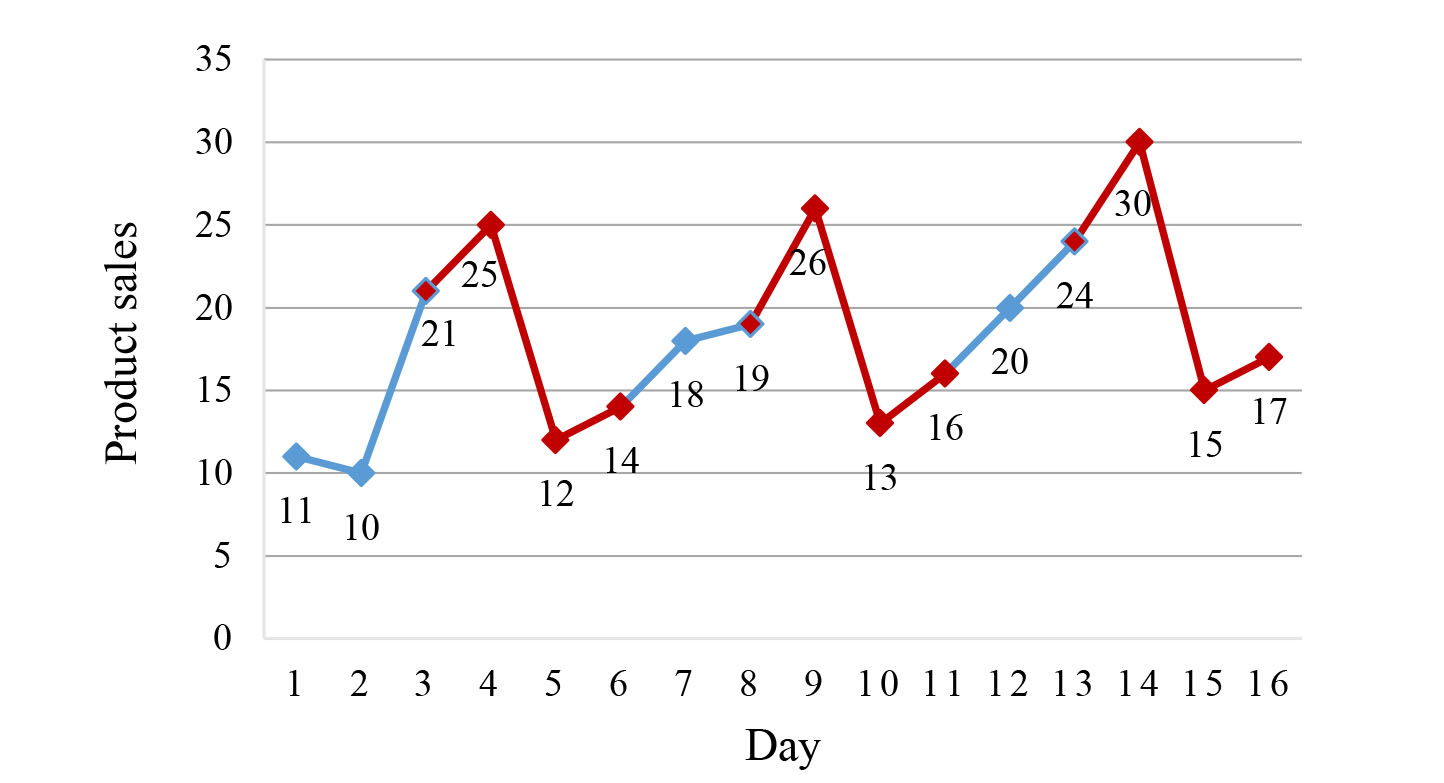}
	\label{FIG:1,intro}
	\caption{For sub-sequence ($ \textit{s}_{3} $, $ \textit{s}_{4} $, $ \textit{s}_{5} $, $ \textit{s}_{6} $) $ = $ (21, 25, 12, 14), 12 is the lowest in this sub-sequence, 14 is the second lowest, 21 is the third lowest number and 25 is the highest one. Therefore, the relative order of sub-sequence ($ \textit{s}_{3} $, $ \textit{s}_{4} $, $ \textit{s}_{5} $, $ \textit{s}_{6} $) is (3, 4, 1, 2). Similarly, the relative orders of sub-sequences ($ \textit{s}_{8} $, $ \textit{s}_{9} $, $ \textit{s}_{10} $, $ \textit{s}_{11} $) and ($ \textit{s}_{13} $, $ \textit{s}_{14} $, $ \textit{s}_{15} $, $ \textit{s}_{16} $) are also (3, 4, 1, 2). (3, 4, 1, 2) is a frequent OPP since there are three sub-sequences, colored in red, with the same relative order (3, 4, 1, 2).}
\end{figure}

	From Fig.~\ref{FIG:1,intro}, (3, 4, 1, 2) is a frequent OPP with length 4 and it occurs three times in the time series $\textbf{s}$. The pattern letters (3, 4, 1, 2) specify their relative orders where 1, 2, 3, and 4 define increasing value, so (3, 4, 1, 2) means that the pattern shifts from the second highest (3) to the highest (4) then to the lowest (1), followed by the second highest (2).
	
	The goal of OPP mining is to find all frequent OPPs from the given time series. In Example 1, there are seven frequent OPPs in time series \textbf{s}: (1, 2), (2, 1), (1, 2, 3), (2, 3, 1), (3, 1, 2), (1, 2, 3, 4), and (3, 4, 1, 2). Obviously, this result set contains many redundant short patterns. For example, for pattern (1, 2), one of its order-preserving super-patterns is (2, 3, 1) which is also frequent and contains the trends that pattern (1, 2) represents. Therefore, to further compress the result set, we also study finding the maximal OPPs whose order-preserving super-patterns are all infrequent. In this way, seven frequent OPPs can be compressed into two maximal OPPs: (1, 2, 3, 4), and (3, 4, 1, 2).

	The main contributions of this paper are as follows. 

\begin{enumerate} 
	\item We study a novel frequent order-preserving pattern discovery problem in time series, and propose an algorithm named OPP-Miner, which consists of two essential tasks, support calculation and candidate pattern generation, for OPP mining.
	
	\item To calculate the support effectively, OPP-Miner adopts filtration and verification strategies. To generate candidate patterns, OPP-Miner employs a pattern fusion strategy that effectively reduces the number of candidate patterns. 
	
	\item To reduce redundant patterns, this paper also proposes the maximal OPP mining method to obtain compression of frequent patterns.
	
	\item We conduct experiments on numerous real-life time series. The results validate that OPP-Miner has good run time efficiency and scalability in handling increasing volumes of time series data. More importantly, experiments show that our algorithms have high utility in analyzing COVID-19 epidemic by identifying critical trends and improve the clustering performance.
\end{enumerate}

	The remainder of the paper is organized as follows. Related work is summarized in Section~\ref{section:Related}, followed by the problem definition in Section~\ref{section:definition}. Section~\ref{section:Algorithm} proposes the OPP-Miner algorithm for mining all frequent OPPs and MOPP-Miner
	for mining maximal ones. Section~\ref{section:Experimental} reports experiments and comparisons, including a case study on finding patterns from COVID-19 infection time series. Section~\ref{section:Conclusion} concludes the paper.

\section{Related work}
\label{section:Related}

	In this section, we review some related work concerning sequential pattern mining, time series representation, and order-preserving pattern matching.
	
\subsection{Sequential pattern mining}
	The traditional sequential pattern mining method refers to mining all frequent patterns. However, the limitations of this method have been gradually revealed since real-life is complicated. Hence, various extension methods have been proposed to meet the different needs of users. For example, various methods, such as closed pattern mining~\cite{Wu2020NetNCSP}, maximal pattern mining~\cite{Yun2014maximal, li2021apinmax}, top-\textit{k} pattern mining~\cite{Huang2019Mining}, and top-rank-\textit{k} pattern mining~\cite{Dam2016efficient}, have been proposed to avoid mining many useless patterns. In addition, traditional sequential pattern mining methods mainly focus on mining frequent patterns. However, this assumption cannot be applied to all situations. Thus, many types of patterns and mining methods have been proposed, such as tri-pattern mining~\cite{Min2020Frequent, wu2022tkdd}, high utility pattern mining~\cite{gan2021tcyb}, contrast pattern mining~\cite{Wang2016Efficient}, rare pattern mining~\cite{Piri2018Development} and co-location pattern mining~\cite{Wang2018co}. Furthermore, sequential pattern mining refers to discovering information from sequential data. Data mining usually uses two types of data, the symbolic sequence and the time series. The symbolic sequence consists of characters, such as DNA and gene sequences. Mining interesting patterns in a symbolic sequence has numerous applications, such as uncertain databases~\cite{Lin2020Utilitye} and customer demand analysis~\cite{Nishimuraa2018latent}. The time series consists of continuous numerical values, such as the daily weather temperature~\cite{Xu2015short}, ECG/EEG data~\cite{Samiee2015Epileptic} and the daily share price~\cite{Li2020multimodal}. Finding critical trends and periodic patterns in time series can be applied to many fields, such as a garden route planning~\cite{Tsai2015location}  and recommend products with higher profits \cite {wugeng2021kbs, wulei2021eswa}.

%dynamic change analysis of stock value~\cite{Wu2014Data},

\subsection{Time series representation}
	Frequent pattern discovery from time series is of great significance and challenge. However, most pattern mining methods find important patterns in symbol sequences, while a time series is usually a series of discrete-time data. To mine patterns in a time series, segmented representation and symbolic representation methods are adopted to transform discrete-time data into other forms. Segmented representation employs the segmented means~\cite{Guo2010improved} or segmented sum of variation (SSV)~\cite{Lee2003Dimensionality} to represent the time series data. For example, Lin \textit{et al}.~\cite{Lin2002Finding} proposed the piecewise aggregate approximation (PAA) method to reduce the dimension of a time series and then mined frequent patterns. Symbolic representation is another universal method since after adopting this method, many pattern mining methods can be applied to a time series. For example, Keogh \textit{et al}.~\cite{Keogh2005HOT} proposed the symbolic aggregate approximation (SAX) method to convert numeric data to a symbolic form and then detected abnormal patterns in the converted sequence. Tan \textit{et al}.~\cite{Tan2016Discovering} adopted fluctuation between contiguous data to transform a time series into a symbol string and then mined the patterns with weak-wildcard gaps. 
	Although these methods can realize the mining of time series, discrete-time data is transformed into other forms, which makes the original data deviate. How to mine frequent patterns in the original time series is a critical issue. 

\subsection{Order-preserving pattern matching}
	Recently, order-preserving pattern matching~\cite{Kim2014Order} has been proposed, which represents the trend in a time series effectively since it employs the order relation to represent the rank of each number in the time series. For example, given a time series \textbf{s} $ = $ (39, 46, 10, 21), the rank of the value from low to high are 1 (10), 2 (21), 3 (39), and 4 (46). Therefore, the time series represented in the ranked order space is (3, 4, 1, 2).

	Initially, the order relation between two numbers was represented in a binary relationship ($ < $ , $ > $)~\cite{Kim2014Order}, meaning that characters with numerical equality are not considered. Noticing this disadvantage, Cho et al.~\cite{Cho2013Fast} extended the binary relation to a ternary relationship ($ > $ , $ < $ , $ = $) and designed a new algorithm to determine whether two strings are order-isomorphic even if some elements are equal. These studies belong to exact order-preserving pattern matching, which requires the relative order of pattern and sub-sequence to be exactly the same. For diverse purposes in time series, many effective methods have been proposed, such as the order-preserving suffix tree~\cite{Crochemore2013Order}, binary conversion filtration~\cite{Chhabra2016filtration}, and SIMD instruction filtration~\cite{Chhabra2016Engineering}. In addition, approximate order-preserving pattern matching variants were proposed to allow data noise. One approximation case is that two strings are matched if they have the same relative order after removing up to \textit{k} elements in the same positions in both strings~\cite{Pawel2016Order}. However, this method cannot measure the local approximation between two strings, which leads to the deviation of matching results. To handle this problem, a similarity measure method based on ($ \delta $ , $ \gamma $) distance~\cite{Juan2018New,Wu2020NetDAP} was proposed, using local and global constraints to improve the matching accuracy.

	Table~\ref{tab1:Related} gives a comparison of the related work.
	
	\begin{table*}
		\centering
		\caption{A summary and comparison of related work on mining frequent patterns from time series.}
		\label{tab1:Related}
		\begin{tabular}{ccccccc}
			\toprule   
			Research&	Matching / Mining&	Support&	Representation type&	Binary / Ternary relation\\
			
			\midrule 
			Lin et al.~\cite{Lin2002Finding}&	Mining	&Exact&	Segmented representation&	-\\
			Keogh et al.~\cite{Keogh2005HOT}&	Mining&	Exact&	Symbolic representation	&-\\
			Tan et al.~\cite{Tan2016Discovering}&	Mining&	Exact&	Symbolic representation&	-\\
			Kim et al.~\cite{Kim2014Order}&	Matching&	Exact&	Order relation representation&	Binary\\
			Cho et al.~\cite{Cho2013Fast}	&Matching&	Exact&	Order relation representation	&Ternary\\
			Paweł et al.~\cite{Pawel2016Order}&	Matching&	Approximate&	Order relation representation&	Ternary\\
			Juan et al.~\cite{Juan2018New}&Matching&	Approximate	&Order relation representation	&Ternary\\
			This paper	&Mining	&Exact&	Order relation representation&	Binary\\
			
			\bottomrule
		\end{tabular}
	\end{table*}

	Inspired by order-preserving pattern matching~\cite{Kim2014Order}, we investigate mining all frequent OPPs and maximal OPPs in a given time series that employs the order-preserving pattern matching to calculate the support of patterns. This research is of more practical significance since people can quickly capture the critical trends in a time series by analysing the discovered frequent OPPs.

\section{Problem definition}
\label{section:definition}

	This section defines the basic concepts and formally introduces the problem of OPP mining.

\begin{definition}
	A time series is composed of ordered continuous values, which can be expressed as \textbf{s} $ = $ ($ \textit{s}_{1} $, \dots, $ \textit{s}_{{j}}$  , \dots,  $\textit{s}_{{n}} $), where \textit{n} is the number of measurements, and $\textit{s}_{{j}}$ (1 $\leq$ $ \textit{j} $ $\leq$ $ \textit{n} $)  is the value measured at time \textit{j}.
\end{definition}

\begin{definition}
	Given time series \textbf{s} $ = $ ($ \textit{s}_{1} $, \dots,  $ \textit{s}_{{j}} $ , \dots,  $ \textit{s}_{{n}} $), the rank of $ \textit{s}_{{j}} $ is the number of $ \textit{s}_{{x}} $ whose value is less than $ \textit{s}_{{j}} $ for all 1 $\leq$ $ \textit{x} $ $\leq$ $ \textit{n} $, denoted as  $ \textit{rank}_{{s}} $($\textit{s}_{{j}} $), is defined as:  $ \textit{rank}_{{s}} $($ \textit{s}_{{j}} $) $ = $ 1 + $ \vert $$ \lbrace $$\textit{x} $  $ \vert $ $ \textit{s}_{{x}} $ $ < $ $ \textit{s}_{{j}} $, 1 $\leq$ $ \textit{x} $ $\leq$ $ \textit{n} $$ \rbrace $$ \vert $.
\end{definition}

\begin{definition}
	Given time series \textbf{s}, the relative order of \textbf{s}, denoted as $ \sigma $(\textbf{s}), is defined as: $ \sigma $(\textbf{s}) $ = $ ($ \textit{rank}_{{s}} $($ \textit{s}_{1} $), $ \textit{rank}_{{s}} $($ \textit{s}_{2} $) , \dots , $ \textit{rank}_{{s}} $($ \textit{s}_{{n}} $)).
\end{definition}

\begin{definition}
	The pattern represented by the relative order of the sequence is called the OPP.
\end{definition}

\begin{definition}
	Let \textbf{s} $ = $ ($ \textit{s}_{1} $, $ \textit{s}_{2} $, \dots , $ \textit{s}_{{n}} $) be a time series and \textbf{p} $ = $ ($ \textit{p}_{1} $, $ \textit{p}_{2} $, \dots , $ \textit{p}_{{m}} $) be an OPP. \textit{L} $ = $ $ < $$ \textit{l}_{1} $, $ \textit{l}_{2} $, \dots , $ \textit{l}_{{m}} $$ > $ is an occurrence of pattern \textbf{p} in \textbf{s}, if and only if $ \sigma $($ \textit{s}_{{l}_{1}} $, $ \textit{s}_{{l}_{2}} $, \dots, $ \textit{s}_{{l}_{{m}}} $) $ = $ ($ \textit{p}_{1} $, $ \textit{p}_{2} $, \dots , $ \textit{p}_{{m}} $). The number of occurrences of \textbf{p} in \textbf{s} is the support, represented by \textit{sup}(\textbf{p}, \textbf{s}).
\end{definition}

\begin{definition}
	If \textit{sup}(\textbf{p}, \textbf{s}) is no less than the minimum support threshold \textit{minsup}, pattern \textbf{p} is called a frequent pattern.
\end{definition}

	\textbf{Problem Statement}. Given time series \textbf{s} and minimum support threshold \textit{minsup}, the problem of OPP mining is to discover all frequent OPPs.

\begin{example}
	In Example \ref{example1}, for sub-sequence $ \textbf{t}_{1} $ $ = $ ($ \textit{s}_{3} $, $ \textit{s}_{4} $, $ \textit{s}_{5} $, $ \textit{s}_{6} $) $ = $ (21, 25, 12, 14), the rank of each number in $ \textbf{t}_{1} $ is 3, 4, 1, and 2, respectively. Hence,  $ \sigma $($ \textbf{t}_{1} $) $ = $ (3, 4, 1, 2). We refer to pattern $ \textbf{p} $ $ = $ (3, 4, 1, 2) as an OPP. The support of \textbf{p} in \textbf{s} is 3, since $ \sigma $($ \textit{s}_{3} $, $ \textit{s}_{4} $, $ \textit{s}_{5} $, $ \textit{s}_{6} $) $ = $ $ \sigma $($ \textit{s}_{8} $, $ \textit{s}_{9} $, $ \textit{s}_{10} $, $ \textit{s}_{11} $) $ = $ $ \sigma $($ \textit{s}_{13} $, $ \textit{s}_{14} $, $ \textit{s}_{15} $, $ \textit{s}_{16} $) $ = $ ($ \textit{p}_{1} $, $ \textit{p}_{2} $, $ \textit{p}_{3} $, $ \textit{p}_{4} $) $ = $ (3, 4, 1, 2), i.e., $ < $3, 4, 5, 6$ > $, $ < $8, 9, 10, 11$ > $, and $ < $13, 14, 15, 16$ > $ are three occurrences of \textbf{p} in \textbf{s}. When \textit{minsup} $ = $ 3, there are seven frequent OPPs in time series \textbf{s}: (1, 2), (2, 1), (1, 2, 3), (2, 3, 1), (3, 1, 2), (1, 2, 3, 4) and (3, 4, 1, 2). 
\end{example}

\section{Algorithm design}
\label{section:Algorithm}
	Mining OPPs requires two essential tasks: (1) pattern support calculation; and (2) candidate pattern generation. The former will check the support (frequency) of the pattern with respect to the underlying time series databases, and the latter will generate a set of candidate patterns.	
	
	In this section, we first propose two strategies, the filtration and verification strategies, to calculate the pattern support in Subsection \ref{subsection:Support}, followed by  Subsection~\ref{subsection:Candidate} which studies the pattern fusion strategy to generate candidate patterns. Subsection~\ref{subsection:OPP-Miner} introduces the OPP-Miner algorithm  for mining all frequent patterns and Subsection~\ref{subsection:MOPP-Miner} studies mining maximal ones.
	
\subsection{Support calculation}
\label{subsection:Support}
	To determine whether a candidate pattern is frequent, it is necessary to calculate its support in sequence~\cite{Wu2020NetNCSP,Wu2017Strict}. In this subsection, we will introduce the support calculation algorithm FVP, which adopts filtration and verification strategies to effectively filter redundant sub-sequences and finds the occurrences satisfying the order relations. The specific methods are elaborated in the following.

	Given time series \textbf{s} $ = $ ($ \textit{s}_{1} $, $ \textit{s}_{2} $, \dots , $ \textit{s}_{{n}} $) and pattern \textbf{p} $ = $ ($ \textit{p}_{1} $, $ \textit{p}_{2} $, \dots , $ \textit{p}_{{m}} $), the process of support calculation is essential to find the sub-sequences with the same relative orders as the pattern. There are five steps to calculate the support.

	Step 1. Formalize the order relations of pattern \textbf{p}. To address this issue, the numbers in pattern \textbf{p} are sorted in ascending order, and the position of the number with rank \textit{i} in pattern \textbf{p} is denoted as $ \textit{index}[\textit{i}] $. The results are stored in an ordered table, where $ \textit{p}_{{index}[{i}]} $ $ < $ $ \textit{p}_{{index}[{i}+1]} $ (1 $\leq$ $ \textit{i} $ $\leq$ $ \textit{m} $). For example, suppose pattern \textbf{p} $ = $ (3, 4, 5, 1, 2). The ordered table of pattern \textbf{p} is shown in Table~\ref{tab2:ordertable}.

\begin{table}[h]
	\centering
	\caption{The ordered table of pattern (3, 4, 5, 1, 2)}
	\label{tab2:ordertable}
	\begin{tabular}{ccccccc}
	\toprule   
		Sorted pattern & 1 & 2 & 3 & 4 & 5\\
	\midrule 
		\textit{Index}[\textit{i}] & 4 & 5 & 1 & 2 & 3\\
	\bottomrule
	\end{tabular}
\end{table}

	Step 2. \textbf{\textit{Filtration Strategy}}. The consecutive numbers in pattern \textbf{p} and time series \textbf{s} are compared pairwise and the results are encoded as \textbf{p$ ’ $} $ = $ ($ \textit{a}_{1} $, \dots, $ \textit{a}_{{i}} $, \dots, $ \textit{a}_{{m}-1} $) and \textbf{s$ ’ $} $ = $ ($ {\textit{b}}_{1} $, \dots , ${\textit{b}}_{j} $, \dots , $ \textit{b}_{{n}-1} $) using binary numbers according to the following equations:
	
	\begin{equation}
	a_{i} = 
	\begin{cases}
	1,& p_{i} < p_{i+1}   (1\leq i \leq m-1)\\\\
	0, & p_{i} > p_{i+1}(1\leq i \leq m-1)
	\end{cases}
	\end{equation}
	
	\begin{equation}
	b_{i} = 
	\begin{cases}
	1,& s_{j} < s_{j+1}   (1\leq j \leq n-1)\\\\
	0, & s_{j} > s_{j+1}(1\leq j \leq n-1)
	\end{cases}
	\end{equation}

	Step 3. \textbf{\textit{Filtration Strategy}}. $ \rm{SBNDM}_{2} $~\cite {Durian2010Improving} is employed to find the same sub-sequence as \textbf{p$ ’ $} in \textbf{s$ ’ $}.

	Step 4. \textbf{\textit{Verification Strategy}}. According to Step 2, a sub-sequence in \textbf{s$ ’ $} is ($ \textit{b}_{l_{1}} $, $ \textit{b}_{l_{2}} $, \dots , $ \textit{b}_{{l}_{{m}-1}} $) whose corresponding sub-sequence in \textbf{s} is ($ \textit{s}_{l_{1}} $, $ \textit{s}_{l_{2}} $, \dots , $ \textit{s}_{{l}_{{m}-1}} $, $ \textit{s}_{{l}_{{m}}} $). We will verify whether ($ \textit{s}_{l_{1}} $, $ \textit{s}_{l_{2}} $, \dots , $ \textit{s}_{{l}_{{m}-1}} $, $ \textit{s}_{{l}_{{m}}} $) can be an occurrence. According to the ordered table, \textit{L} $ = $ $ < $$ \textit{l}_{1} $, $ \textit{l}_{2} $, \dots , $ \textit{l}_{{m}} $$ > $ is an occurrence of pattern \textbf{p} in time series \textbf{s}, if and only if $ \textit{s}_{l_{1}-1+index[i]} $ $ < $ $ \textit{s}_{l_{1}-1+index[i+1]} $ for all 1 $\leq$ \textit{i} $\leq$ \textit{m}-1.

	Step 5. Iterate Steps 3 and 4 until all sub-sequences are found and verified.

	An illustrative example is shown as follows.

\begin{example}
	In this example, we select the time series in Example \ref{example1} and pattern \textbf{p} $ = $ (3, 4, 5, 1, 2). 
	
	Step 1. According to pattern \textbf{p}, we create the ordered table of pattern \textbf{p} shown in Table~\ref{tab2:ordertable}. 
	
	Step 2. Pattern \textbf{p} and time series \textbf{s} are converted into binary number strings, where \textbf{p}$ ’ $ $ = $ (1, 1, 0, 1) and \textbf{s}$ ’ $ $ = $ ($ \textit{b}_{1} $, $ \textit{b}_{2} $, $ \textit{b}_{3} $, $ \textit{b}_{4} $, $ \textit{b}_{5} $, $ \textit{b}_{6} $, $ \textit{b}_{7} $, $ \textit{b}_{8} $, $ \textit{b}_{9} $, $ \textit{b}_{10} $, $ \textit{b}_{11} $, $ \textit{b}_{12} $, $ \textit{b}_{13} $, $ \textit{b}_{14} $, $ \textit{b}_{15} $) $ = $ (0, 1, 1, 0, 1, 1, 1, 1, 0, 1, 1, 1, 1, 0, 1), respectively.
	
	Step 3. $ \rm{SBNDM}_{2} $ is employed to find the same sub-sequence as \textbf{p}$ ’ $ in \textbf{s}$ ’ $, and the first sub-sequence found is ($ \textit{b}_{2} $, $ \textit{b}_{3} $, $ \textit{b}_{4} $, $ \textit{b}_{5} $).
	
	Step 4. Sub-sequence ($ \textit{b}_{2} $, $ \textit{b}_{3} $, $ \textit{b}_{4} $, $ \textit{b}_{5} $) in \textbf{s}$ {’} $ corresponds to sub-sequence ($ \textit{s}_{2} $, $ \textit{s}_{3} $, $ \textit{s}_{4} $, $ \textit{s}_{5} $, $ \textit{s}_{6} $) in \textbf{s}. When \textit{i} $ = $ 1, we verify whether $ \textit{s}_{l_{1}-1+index[1]} $ $ < $ $ \textit{s}_{l_{1}-1+index[2]} $ is correct. According to Table~\ref{tab2:ordertable}, we know that \textit{index}[1] and \textit{index}[2] are 4 and 5, respectively. Therefore, we verify whether $ \textit{s}_{2-1+4} $ $ < $ $ \textit{s}_{2-1+5} $ is correct. $ \textit{s}_{2-1+4} $ $ < $ $ \textit{s}_{2-1+5} $ is correct since $ \textit{s}_{5} $ and $ \textit{s}_{6} $ are 12 and 14, respectively. Similarly, when \textit{i} $ = $ 2, we verify whether $ \textit{s}_{l_{1}-1+index[2]} $ $ < $ $ \textit{s}_{l_{1}-1+index[3]} $ is correct. $ \textit{s}_{2-1+5} $ $ < $ $ \textit{s}_{2-1+1} $ is incorrect since $ \textit{s}_{6} $ and $ \textit{s}_{2} $ are 14 and 10, respectively. Thus, ($ \textit{s}_{2} $, $ \textit{s}_{3} $, $ \textit{s}_{4} $, $ \textit{s}_{5} $, $ \textit{s}_{6} $) is not an occurrence of pattern \textbf{p} in time series \textbf{s}.

	Step 5. Iterate Steps 3 and 4, we know that $ < $7, 8, 9, 10, 11$ > $ and $ < $12, 13, 14, 15, 16$ > $ are two occurrences of \textbf{p} in \textbf{s}. Hence, the support of pattern \textbf{p} is 2.
\end{example}

	The FVP algorithm is shown in Algorithm \ref {alg:fvp}.
	
		\begin{algorithm}[h] 
		\caption{FVP: Filtration and Verification Algorithm}  \label{alg:fvp}
		\hspace*{0.02in} {\textbf{Input:} time series \textbf{s}, pattern \textbf{p}}\\	
		\hspace*{0.02in} {\bf Output:}
		\textit{sup}(\textbf{p}, \textbf{s})		
		\begin{algorithmic}[1]
			\State Create the ordered table of \textbf{p};
			\State $ \textbf{p}’ $ $ \leftarrow $  \textit{transform}_\textit{pat}(\textbf{p}) according to Equation (1);
			\State $ \textbf{s}’ $ $ \leftarrow $  \textit{transform}_\textit{pat}(\textbf{s}) according to Equation (2);
			\While{$ \rm{SBNDM}_{2} $($ \textbf{p}’ $, $ \textbf{s}’ $) $ \neq $ null}
			\State \textit{occ} $ \leftarrow $ $ \rm{SBNDM}_{2} $($ \textbf{p}’ $, $ \textbf{s}’ $);
			\If{\textit{occ} is an occurrence according to the verification strategy} \textit{sup}(\textbf{p}, \textbf{s}) ++ ;
			\EndIf
			\EndWhile
			\State return \textit{sup}(\textbf{p}, \textbf{s});
		\end{algorithmic}
	\end{algorithm}

	It is worth noting that we can directly adopt the verification strategy to find the occurrences of \textbf{p} in \textbf{s}. It is very slow since every sub-sequence with equal length of \textbf{p} should be verified. Therefore, to improve the efficiency, we add a filtration phase before the verification phase, to filter redundant sub-sequences, which can greatly reduce the verification times.

\subsection{Candidate pattern generation}
\label{subsection:Candidate}
	Since OPP is represented by relative order, traditional candidate pattern generation methods cannot be directly applied to it. Therefore, to generate candidate OPPs, we first propose an enumeration strategy, and then propose a more efficient pattern fusion strategy.

\subsubsection{Enumeration strategy}
	A frequent OPP \textbf{p} $ = $ ($ \textit{p}_{1} $, \dots, $ \textit{p}_{{j}}$  , \dots,  $\textit{p}_{{m}} $) with length \textit{m}, where 1 $\leq$ \textit{j} $\leq$ \textit{m}, can generate \textit{m} $ + $ 1 kinds of candidate patterns with length \textit{m} $ + $ 1, which are $ \textbf{t}_{\textbf{1}} $, \dots, $ \textbf{t}_{{\textbf{\textit{i}}}}$  , \dots,  $\textbf{t}_{{\textbf{\textit{m}+1}}} $, respectively. To maintain the relative order of pattern \textbf{p}, when generating candidate pattern $ \textbf{t}_{\textbf{\textit{i}}} $ $ = $ ($ t^{1}_{i} $, $ t^{2}_{i} $, \dots , $ t^{m}_{i} $, \textit{i}), if $ \textit{p}_{j} $ $\geq$ \textit{i}, then $ t^{j}_{i} $ $ = $ $ \textit{p}_{j} $ $ + $ 1, otherwise $ t^{j}_{i} $ $ = $ $ \textit{p}_{j} $, for all 1 $\leq$ \textit{i} $\leq$ \textit{m} $ + $ 1.

\begin{example}
\label{example4}
	Suppose there are three frequent OPPs with length 3: (1, 2, 3), (2, 3, 1) and (3, 1, 2). Taking pattern \textbf{p} $ = $ (1, 2, 3) as an example, the enumeration strategy generates four candidate patterns, $ \textbf{t}_{\textbf{1}} $, $ \textbf{t}_{\textbf{2}} $, $ \textbf{t}_{\textbf{3}} $ and $ \textbf{t}_{\textbf{4}} $, respectively. For $ \textbf{t}_{\textbf{1}} $ $ = $ ($ t^{1}_{1} $, $ t^{2}_{1} $, $ t^{3}_{1} $, $ t^{4}_{1} $), first, let $ t^{4}_{1} $ be 1. Then, the relative order of the first three numbers of $\textbf{t}_{\textbf{1}} $ should be the same as pattern \textbf{p}. Since $ \textit{p}_{1} $, $ \textit{p}_{2} $ and $ \textit{p}_{3} $ are all no less than $ t^{4}_{1} $, thus $ t^{1}_{1} $, $ t^{2}_{1} $ and $ t^{3}_{1} $ are 2, 3 and 4, respectively, i.e. $ t^{1}_{1} $ $ = $ $ \textit{p}_{1} $ $ + $ 1 $ = $ 2, $ t^{2}_{1}  = $ $ \textit{p}_{2} $ $ + $ 1 $ = $ 3, $ t^{3}_{1} $ $ = $ $ \textit{p}_{3} $ $ + $ 1 $ = $ 4. Finally, we obtain $ \textbf{t}_{\textbf{1}}  = $ (2, 3, 4, 1). Similarly, we can get $ \textbf{t}_{\textbf{2}} $ $ = $ (1, 3, 4, 2), $\textbf{t}_{\textbf{3}} $ $ = $ (1, 2, 4, 3), and $\textbf{t}_{\textbf{4}} $ $ = $ (1, 2, 3, 4). Other patterns are processed in the same way. The results are shown in Table \ref{tab3:enumeration}. 
\end{example}

\begin{table}[H]
	\centering
	\caption{Generate candidate OPPs adopting the enumeration strategy}
	\label{tab3:enumeration}
	\begin{tabular}{ccccccc}
		\toprule   
		Frequent OPP & Candidate OPPs\\
		\midrule 
		(1, 2, 3) & (2, 3, 4, 1), (1, 3, 4, 2), (1, 2, 4, 3), (1, 2, 3, 4)\\
		(2, 3, 1) & (3, 4, 2, 1), (3, 4, 1, 2), (2, 4, 1, 3), (2, 3, 1, 4)\\
		(3, 1, 2) & (4, 2, 3, 1), (4, 1, 3, 2), (4, 1, 2, 3), (3, 1, 2, 4)\\		
		\bottomrule
	\end{tabular}
\end{table}

	It is worth noticing that under the enumeration strategy, $ \textit{n}! $ candidate patterns with length \textit{n} will be generated in the worst case. Obviously, when \textit{n} increases, the number of candidate patterns will be very large, which will reduce the algorithm efficiency. Therefore, a more efficient method is needed to reduce the number of candidate patterns.

\subsubsection{Pattern fusion strategy}
	To handle the disadvantage of too many candidate patterns generated by the enumeration strategy, in this subsection, we propose a pattern fusion strategy that can effectively reduce the generation of redundant patterns.

\begin{definition}
Given pattern \textbf{p} $ = $ ($ \textit{p}_{1} $, $ \textit{p}_{2} $, \dots , $ \textit{p}_{{m}} $), sub-sequence ($ \textit{p}_{1} $, $ \textit{p}_{2} $, \dots , $ \textit{p}_{{m-1}} $) is called the prefix pattern of \textbf{p}, denoted by \textit{prefix}(\textbf{p}). Sub-sequence \textbf{p} $ = $ ($ \textit{p}_{2} $, \dots , $ \textit{p}_{{m-1}} $, $ \textit{p}_{{m}} $) is called the suffix pattern of \textbf{p}, denoted by \textit{suffix}(\textbf{p}). 
\end{definition}

\begin{definition}
	Given pattern \textbf{p}, the relative order of \textit{prefix}(\textbf{p}) is called the order-preserving prefix pattern, denoted by \textit{prefixorder}(\textbf{p}). The relative order of \textit{suffix}(\textbf{p}) is called the order-preserving suffix pattern, denoted by \textit{suffixorder}(\textbf{p}).
\end{definition}

	\textbf{\textit{Pattern Fusion Strategy}}. For frequent OPPs \textbf{p} $ = $ ($ \textit{p}_{1} $, $ \textit{p}_{2} $, \dots , $ \textit{p}_{{m}} $) and \textbf{q} $ = $ ($ \textit{q}_{1} $, $ \textit{q}_{2} $, \dots , $ \textit{q}_{{m}} $) with length \textit{m}, if they are fused to generate candidate patterns with length \textit{m} $ + $ 1, there are two cases, the general case and the special case.

	(1) General case: if \textit{suffixorder}(\textbf{p}) $ = $ \textit{prefixorder}(\textbf{q}) but \textit{suffix}(\textbf{p}) $ \ne $ \textit{prefix}(\textbf{q}), then \textbf{p} and \textbf{q} can be fused into one candidate pattern \textbf{x} $ = $ ($ \textit{x}_{1} $, $ \textit{x}_{2} $, \dots , $ \textit{x}_{{m+1}} $). The fusion rules are as follows.
	
	If $ \textit{p}_{1} $ $ < $ $ \textit{q}_{m} $, then $ \textit{x}_{1} $ $ = $ $ \textit{p}_{1} $ and $ \textit{x}_{m + 1} $ $ = $ $ \textit{q}_{m} $ $ + $ 1, and if $ \textit{p}_{u} $ $ > $ $ \textit{q}_{m} $, then $ \textit{x}_{u} $ $ = $ $ \textit{p}_{u} $ $ + $ 1, otherwise $ \textit{x}_{u} $ $ = $ $ \textit{p}_{u} $, for all 2 $\leq$ \textit{u} $\leq$ \textit{m}.
	
	If $ \textit{p}_{1} $ $ > $ $ \textit{q}_{m} $, then $ \textit{x}_{1} $ $ = $ $ \textit{p}_{1} $ $ + $ 1 and $ \textit{x}_{m + 1} $ $ = $ $ \textit{q}_{m} $, and if $ \textit{q}_{v} $ $ > $ $ \textit{p}_{1} $, then $ \textit{x}_{v+1} $ $ = $ $ \textit{q}_{v} $ $ + $ 1, otherwise $ \textit{x}_{v+1} $ $ = $ $ \textit{q}_{v} $, for all 1 $\leq$ \textit{v} $\leq$ \textit{m} $ - $ 1.

	(2) Special case: if \textit{suffixorder}(\textbf{p}) $ = $ \textit{prefixorder}(\textbf{q}) and \textit{suffix}(\textbf{p}) $ = $ \textit{prefix}(\textbf{q}), then \textbf{p} and \textbf{q} can be fused into two candidate patterns, \textbf{y} $ = $ ($ \textit{y}_{1} $, $ \textit{y}_{2} $, \dots , $ \textit{y}_{{m+1}} $) and \textbf{k} $ = $ ($ \textit{k}_{1} $, $ \textit{k}_{2} $, \dots , $ \textit{k}_{{m+1}} $). The fusion rules are as follows.
	
	For pattern \textbf{y}, let $ \textit{y}_{1} $ $ = $ $ \textit{p}_{1} $ $ + $ 1 and $ \textit{y}_{{m}+1} $ $ = $ $ \textit{p}_{m} $, and if $ \textit{p}_{u} $ $ > $ $ \textit{p}_{1} $, then $ \textit{y}_{u} $ = $ \textit{p}_{u} $ $ + $ 1, otherwise $ \textit{y}_{u} $ $ = $ $ \textit{p}_{u} $, for all 2 $\leq$ \textit{u} $\leq$ \textit{m}. 
	
	For pattern \textbf{k}, let $ \textit{k}_{1} $ $ = $ $ \textit{p}_{1} $ and $ \textit{k}_{{m}+1} $ $ = $ $ \textit{p}_{m} $ $ + $ 1, and if $ \textit{p}_{u} $ $ > $ $ \textit{p}_{1} $, then $ \textit{k}_{u} $ = $ \textit{p}_{u} $ $ + $ 1, otherwise $ \textit{k}_{u} $ $ = $ $ \textit{p}_{u} $, for all 2 $\leq$ \textit{u} $\leq$ \textit{m}.

	The following example illustrates how to generate candidate patterns based on the above two cases.

\begin{example}
	Given pattern \textbf{p} $ = $ (2, 1, 3, 4), let it fuse with pattern \textbf{q} $ = $ (1, 2, 3, 4) and \textbf{r} $ = $ (1, 3, 4, 2) respectively to generate candidate patterns. 
	
	Firstly, we know that \textit{suffix}(\textbf{p}) = (1, 3, 4), \textit{suffixorder}(\textbf{p}) $ = $ (1, 2, 3), \textit{prefix}(\textbf{q}) $ = $ (1, 2, 3), \textit{prefixorder}(\textbf{q}) $ = $ (1, 2, 3), \textit{prefix}(\textbf{r}) $ = $ (1, 3, 4), and \textit{prefixorder}(\textbf{r}) $ = $ (1, 2, 3). 
	
	1) Since \textit{suffixorder}(\textbf{p}) $ = $ \textit{prefixorder}(\textbf{q}) but \textit{suffix}(\textbf{p}) $ \ne $ \textit{prefix}(\textbf{q}), which satisfies the general case, \textbf{p} and \textbf{q} can be fused into one candidate pattern \textbf{x}. Since $ \textit{p}_{1} $ $ < $ $ \textit{q}_{4} $, $ \textit{x}_{1} $ $ = $ $ \textit{p}_{1} $ $ = $ 2 and $ \textit{x}_{5} $ $ = $ $ \textit{q}_{4} $ $ + $ 1 $ = $ 5. Since $ \textit{p}_{2} $, $ \textit{p}_{3} $, and $ \textit{p}_{4} $ are not greater than $ \textit{q}_{4} $, ($ \textit{x}_{2} $, $ \textit{x}_{3} $, $ \textit{x}_{{4}} $) $ = $ ($ \textit{p}_{2} $, $ \textit{p}_{3} $, $ \textit{p}_{{4}} $) $ = $ (1, 3, 4). Finally, we obtain \textbf{x} $ = $ (2, 1, 3, 4, 5).
	
	2) Since \textit{suffixorder}(\textbf{p}) $ = $ \textit{prefixorder}(\textbf{r}) and \textit{suffix}(\textbf{p}) $ = $ \textit{prefix}(\textbf{r}), which satisfies the special case, \textbf{p} and \textbf{r} can be fused into two candidate patterns: \textbf{t} and \textbf{k}. For pattern \textbf{t}, let $ \textit{t}_{1} $ $ = $ $ \textit{p}_{1} $ $ + $ 1 $ = $ 3, and $ \textit{t}_{5} $ $ = $ $ \textit{p}_{1} $ $ = $ 2. Since $ \textit{p}_{2} $ $ < $ $ \textit{p}_{1} $, $ \textit{t}_{2} $ $ = $ $ \textit{p}_{2} $ $ = $ 1. Since $ \textit{p}_{3} $ and $ \textit{p}_{4} $ are both greater than $ \textit{p}_{1} $, $ \textit{t}_{3} $ $ = $ $ \textit{p}_{3} $ $ + $ 1 $ = $ 4 and $ \textit{t}_{4} $ $ = $ $ \textit{p}_{4} $ $ + $ 1 $ = $ 5. Finally, we get \textbf{t} = (3, 1, 4, 5, 2). For pattern \textbf{k}, let $ \textit{k}_{1} $ $ = $ $ \textit{p}_{1} $ $ = $ 2 and $ \textit{k}_{5} $ $ = $ $ \textit{p}_{1} $ $ + $ 1 $ = $ 3, the remaining steps are the same as the generation of \textbf{t}. Finally, \textbf{k} $ = $ (2, 1, 4, 5, 3) can be obtained.
\end{example}

	The pattern fusion strategy outperforms the enumeration strategy. An illustrative example is shown as follows.

\begin{example}
	In this example, we use the same frequent OPPs with length 3 as in Example~\ref{example4}, and adopt the pattern fusion strategy to generate candidate OPPs. For pattern \textbf{$\textbf{p}$} $ = $ (1, 2, 3), it is easy to obtain \textit{suffix}(\textbf{$\textbf{p}$}) $ = $ (2, 3) and \textit{suffixorder}(\textbf{$\textbf{p}$}) $ = $ (1, 2). Then traverse all frequent OPPs with length 3 to find patterns whose relative order of prefix pattern is also (1, 2), and fuse these patterns with \textbf{p} according to the pattern fusion strategy. Finally, pattern $\textbf{p}$ is fused with (1, 2, 3) and (2, 3, 1) respectively, and generates three candidate patterns, which are (1, 2, 3, 4), (2, 3, 4, 1), and (1, 3, 4, 2). Other patterns are processed in the same way. The results are shown in Table \ref{tab4:fusion}.
\end{example}

\begin{table}[H]
	\centering
	\caption{Generate candidate patterns adopting  pattern fusion strategy}
	\label{tab4:fusion}
	\begin{tabular}{ccccccc}
		\toprule   
		Frequent OPP & Candidate OPPs\\
		\midrule 
		(1, 2, 3) & (1, 2, 3, 4), (2, 3, 4, 1), (1, 3, 4, 2)\\
		(2, 3, 1) & (3, 4, 1, 2), (2, 4, 1, 3)\\
		(3, 1, 2) & (4, 1, 2, 3), (3, 1, 2, 4), (4, 2, 3, 1)\\		
		\bottomrule
	\end{tabular}
\end{table}
	
	By comparing the results in Tables \ref{tab3:enumeration} and \ref{tab4:fusion}, it is clear that the pattern fusion strategy can effectively reduce the generation of redundant patterns. The enumeration strategy generates 3 $ \times $ 4 $ = $ 12 candidate patterns since each frequent OPP can generate four kinds of candidate patterns. However, the pattern fusion strategy only generates eight candidate patterns. Therefore, the pattern fusion strategy outperforms the enumeration strategy.
	
	Algorithm \ref {alg:Pf} sketches the PFusion algorithm that generates candidate patterns using the pattern fusion strategy.
	
	\begin{algorithm}[h] 
		\caption{PFusion} \label{alg:Pf}
		\hspace*{0.02in} \textbf{Input:} pattern \textbf{p}, pattern \textbf{q}\\	
		\hspace*{0.02in} {\bf Output:}
		candidate pattern \textbf{c}		
		\begin{algorithmic}[1]
			\State Calculate \textit{suffix}(\textbf{p}) and \textit{suffixorder}(\textbf{p});
			\State Calculate \textit{prefix}(\textbf{q}) and \textit{prefixorder}(\textbf{q});
			\State Fuse \textbf{p} and \textbf{q} into \textbf{c} according to the pattern fusion strategy;
			\State return \textbf{c};
		\end{algorithmic}
	\end{algorithm}

\subsection{OPP-Miner: Mining all frequent OPPs}
\label{subsection:OPP-Miner}
	In this subsection, we propose the OPP-Miner algorithm and analyze its theoretical properties (including time complexity).
	
	Algorithm \ref{alg:opp} sketches the OPP-Miner algorithm that discovers all frequent OPPs. Given time series \textbf{s} and minimum support threshold \textit{minsup}, we define global variables \textit{F} to store all length frequent OPPs. An OPP with length 1 is meaningless because its trend is only one point. Therefore, we start with the OPP with length 2. Firstly, we need to calculate the support of each pattern of \{(1, 2), (2, 1)\} and store the frequent patterns into \textit{F}. The remaining steps of the OPP-Miner algorithm are shown as follows.
\begin{enumerate}[Step 1:]
	\item Extract pattern \textbf{p} in \textit{F};
	\item Traverse pattern \textbf{q} with the same length as \textbf{p} in \textit{F}, and fuse \textbf{p} and \textbf{q} into candidate pattern \textbf{c} according to the pattern fusion strategy.
	\item Calculate the support of pattern \textbf{c}.
	\item If pattern \textbf{c} is frequent, then store it in \textit{F}.
	\item Repeat the above steps until all patterns in \textit{F} are processed.
\end{enumerate}

		\begin{algorithm}[H] 
		\caption{OPP-Miner: OPP Pattern Mining}\label{alg:opp}
		\hspace*{0.02in} \textbf{Input:} time series \textbf{s}, minimum support threshold \textit{minsup}\\	
		\hspace*{0.02in} {\bf Output:}
		frequent OPP set \textit{F}		
		\begin{algorithmic}[1]
			\State Scan \textbf{s} and calculate the support of each pattern of {(1, 2), (2, 1)}, and store the frequent patterns in \textit{F};
			\For{each \textbf{p} in \textit{F}}
			\For{each \textbf{q} in \textit{F} with the same length as \textbf{p}}
			\State \textbf{c} $ \leftarrow $ PFusion(\textbf{p}, \textbf{q});
			\State \textit{support} $ \leftarrow $ FVP(\textbf{s}, \textbf{c});
			\If{\textit{support} $ \geq $ \textit{minsup}}
			\State \textit{F} $ \leftarrow $ \textit{F} $ \cup $ \textbf{c}
			\EndIf
			\EndFor
			\EndFor
			\State return \textit{F};
		\end{algorithmic}
	\end{algorithm}

\begin{theorem}
	OPP satisfies the Apriori property.
\end{theorem}

\begin{proof}
	It can be easily obtained that \textit{sup}(\textit{prefixorder}(\textbf{p}), \textbf{s}) $\geq$ \textit{sup}(\textbf{p}, \textbf{s}) and \textit{sup}(\textit{suffixorder}(\textbf{p}), \textbf{s}) $\leq$ \textit{sup}(\textbf{p}, \textbf{s}) according to Definition 5. If \textit{prefixorder}(\textbf{p}) is not a frequent pattern, i.e. , \textit{sup}(\textit{prefixorder}(\textbf{p}), \textbf{s}) $ < $ \textit{minsup}, then \textit{sup}(\textbf{p}, \textbf{s}) is less than \textit{minsup}. Hence, \textbf{p} is not a frequent pattern either. Similarly, if \textit{suffixorder}(\textbf{p}) is not a frequent pattern, then \textbf{p} is not a frequent pattern either. Therefore, OPP satisfies the Apriori property.
\end{proof}

\begin{theorem}
	The space complexity of OPP-Miner is \textit{O}(\textit{m} $ \times $ (\textit{L} $ + $ \textit{n})), where \textit{m}, \textit{n} and \textit{L} represent the pattern length, the time series length and the number of candidate patterns, respectively.
\end{theorem}

\begin{proof}
	OPP-Miner algorithm space is composed of two parts, the space of frequent and candidate patterns and the space of FVP. It is easy to know that the space complexity of the first part is \textit{O}(\textit{m} $ \times $ \textit{L}). Meanwhile, there are no more than \textit{n} $ - $ \textit{m} occurrences and each occurrence uses \textit{m} spaces. Hence, the space complexity of computing support is \textit{O}(\textit{m} $ \times $ \textit{n}). Therefore, the space complexity of OPP-Miner is \textit{O}(\textit{m} $ \times $ \textit{L} $ + $ \textit{m} $ \times $ \textit{n}) $ = $ \textit{O}(\textit{m} $ \times $ (\textit{L} $ + $ \textit{n})).
\end{proof}

\begin{theorem}
	The time complexity of OPP-Miner is \textit{O}(\textit{m} $ \times $ \textit{n} $ \times $ \textit{L}).
\end{theorem}

\begin{proof}
	Creating the ordered table of \textbf{p} requires \textit{O}(\textit{m} $ \times $ \textit{m}). The filtration strategy finds no more than \textit{n} $ - $ \textit{m} sub-sequences. Each verification costs \textit{O}(\textit{m}). Therefore, the time complexity of FVP is \textit{O}(\textit{m} $ \times $ \textit{m}) $ + $ \textit{O}((\textit{n} $ - $ \textit{m}) $ \times $ \textit{m}) $ = $ \textit{O}(\textit{n} $ \times $ \textit{m}). The time complexity of generating candidate patterns is \textit{O}(\textit{L} $ \times $ \textit{L}). Since FVP runs \textit{L} times, the time complexity of OPP-Miner is \textit{O}(\textit{m} $ \times $ \textit{n} $ \times $ \textit{L} $ + $ \textit{L} $ \times $ \textit{L}) $ = $ \textit{O}(\textit{m} $ \times $ \textit{n} $ \times $ \textit{L}).
\end{proof}

\subsection{MOPP-Miner: Mining maximal OPPs}
\label{subsection:MOPP-Miner}
	For large databases, the number of frequent patterns could be very large. To further compress the frequent patterns, we also propose the MOPP-Miner algorithm which adopts the maximal checking strategy to rule out the redundant patterns.
	
	\begin{definition}
		Given two OPPs $ \textbf{p}_{{1}} $ $ = $ ($ \textit{a}_{1} $, $ \textit{a}_{2} $, \dots , $ \textit{a}_{{i}} $) and $ \textbf{p}_{2} $ $ = $ ($ \textit{b}_{1} $, $ \textit{b}_{2} $, \dots, $ \textit{b}_{{j}} $), $ \textbf{p}_{{1}} $ is an order-preserving sub-pattern of $ \textbf{p}_{{2}} $, if and only if there exist integers $ \textit{k}_{1} $, $ \textit{k}_{2} $, \dots , $ \textit{k}_{{i}} $ such that: (1) 1 $\leq$ $ \textit{k}_{1} $ $ < $ $ \textit{k}_{2} $ $ < $ $ \dots $ $ < $ $ \textit{k}_{{i}} $ $\leq$ \textit{j}; and (2) $ \sigma $($ \textit{a}_{1} $, $ \textit{a}_{2} $, \dots, $ \textit{a}_{{i}} $) $ = $ $ \sigma $($ \textit{b}_{{k}_{1}} $, $ \textit{b}_{{k}_{2}} $, \dots , $\textit{b}_{{k}_{{i}}} $). We also call $ \textbf{p}_{{2}} $  an order-preserving super-pattern of $ \textbf{p}_{{1}} $.
	\end{definition}

	\begin{definition} \label{mopp}
	If pattern \textbf{p} is a frequent OPP and all its order-preserving super-patterns are infrequent, then \textbf{p} is a maximal OPP; otherwise, \textbf{p} is an unmaximal OPP.
	\end{definition}

	\begin{example}
			For pattern (1, 2), one of its order-preserving super-patterns is (2, 3, 1) which is also frequent. Hence, (1, 2) is not a maximal pattern. In Example \ref{example1}, we can find two maximal OPPs which are (1, 2, 3, 4), and (3, 4, 1, 2).
	\end{example}
	
	\textbf{\textit{Maximal Checking Strategy}}. If pattern \textbf{p} is a frequent pattern, neither of the two patterns that are fused to generate \textbf{p} is the maximal pattern. We only need to record the index of these non-maximal patterns. After checking all frequent patterns, the unrecorded patterns  are maximal OPPs.
	
\begin{example}
	According to the pattern fusion strategy, pattern (4, 2, 3, 1) and (3, 4, 1, 2) can be fused into (5, 3, 4, 1, 2). If (5, 3, 4, 1, 2) is frequent, then we can determine that (4, 2, 3, 1) and (3, 4, 1, 2) are not the maximal pattern. 
\end{example}

The MOPP-Miner algorithm embeds the maximal checking strategy into OPP-Miner. Hence, the two algorithms are very similar. Compared with OPP-Miner, MOPP-Miner adds two steps. 

1. In Algorithm \ref {alg:opp},  suppose pattern \textbf{c} is a frequent OPP, according to Defintion \ref{mopp}, order-preserving sub-patterns \textbf{p} and \textbf{q} are not the maximal OPPs and store in \textit{G};

2. After finding all frequent OPPs, if a frequent pattern is not stored in \textit{G}, then it is a maximal OPP.

%After a candidate pattern is generated adopting the pattern fusion strategy and its support is calculated, if it is a frequent OPP, the indexes of its order-preserving prefix pattern and order-preserving suffix pattern will be recorded. After all frequent OPPs are mined, patterns that have not been recorded are the maximal OPPs.

\section{Experimental analysis}
\label{section:Experimental}
	In this section, we conduct experiments and comparisons to verify the following claims: (1) OPP-Miner has good run time efficiency and scalability against the increase in data size; (2) MOPP-Miner can compress the patterns effectively; (3) OPP-Miner can find similar sub-sequences in time series; (4) Our algorithms have high utility in analyzing COVID-19 epidemic by identifying critical trends and improve the clustering performance.  The algorithms and data can be downloaded from \url{https://github.com/wuc567/Pattern-Mining/tree/master/OPP-Miner}.

\subsection{Benchmark Datasets and Baseline Methods}
	All experiments were carried out on an Intel(R) Core(TM) i5-3120M, 2.50GHZ CPU with 8GB RAM and 64-bit Windows 7, using VC++6.0 as the program development environment. To verify the performance of our mining algorithms, this paper adopts different types of real-life time series datasets as experimental data. The specific datasets are summarized in Table~\ref{tab5:dataset}.
	
	\begin{table}[H]
		\centering
		\caption{Benchmark datasets}
		\label{tab5:dataset}
		\begin{tabular}{ccccccc}
			\toprule   
			Dataset & From & Type & Length\\
			\midrule 
			SDB1\textsuperscript{1}  & Russell 2000	& Stock	& 8141\\
			SDB2 &	Nasdaq &	Stock & 12279\\
			SDB3 &	S$ \& $P 500 &	Stock & 23046\\
			SDB4 &	Dow 30 &	Stock &	31387\\
			SDB5\textsuperscript{2} &	Changping &	Temperature &	35064\\
			SDB6 &	Huairou &	Temperature &	35064\\
			SDB7 &	Shunyi &	Temperature &	35064\\
			SDB8 &	Tiantan &	Temperature &	35064\\
			SDB9\textsuperscript{3} &	FreezerRegularTrain	& Sensor &	903000\\
			SDB10\textsuperscript{4} &	Data-Stock &	Stock &	272\\
			SDB11\textsuperscript{5} &	CSSE COVID-19 Dataset &	Daily new cases &	181\\
			SDB12\textsuperscript{6} &	Car &	Sensor &	577\\
			SDB13 &	Beef &	Spectro &	470\\
			\bottomrule
		\end{tabular}
	\flushleft
	\begin{enumerate}[Note1:] 
		\item SDB1-4 databases are the daily values of the stock index and can be downloaded from \url{https://www.yahoo.com/}.
		\item  SDB5-8 databases are from Beijing multi-site air-quality data which is used in Reference~\cite{Zhang2017Cautionary} and can be downloaded from \url{https://archive.ics.uci.edu/ml/datasets.php}.
		\item FreezerRegularTrain is used in Reference~\cite{Murray2015Data} and can be downloaded from \url{http://www.timeseriesclassification.com/index.php}.
		\item Data-Stock is used in Reference~\cite{Tan2016Discovering} and can be downloaded from \url{http://www.fansmale.com/index.html}.
		\item CSSE COVID-19 Dataset is from \url{https://coronavirus.jhu.edu/}.
		\item SDB12 and SDB13 can be downloaded from \url{http://www.timeseriesclassification.com/index.php}, and SDB13 is used in Reference~\cite{Al-Jowder2002Detection}.
	\end{enumerate}
	\end{table}

	In this paper, OPP-Bndm, OPP-Nofilting, OPP-Df and OPP-Bf are employed as the competitive algorithms whose principles are shown as follows.

\begin{enumerate} 
	\item OPP-Bndm and OPP-Nofilting: To analyse the effect of FVP, OPP-Bndm and OPP-Nofilting are proposed. OPP-Bndm employs the classic $ \rm{BNDM} $ algorithm in the filtration strategy to find the same sub-sequences as $ \textbf{p}’ $ in $ \textbf{s}’ $. OPP-Nofilting does not employ the filtration strategy but directly executes the verification strategy.
	
	\item OPP-Df and OPP-Bf: To analyse the effect of pattern fusion strategy, OPP-Df and OPP-Bf are proposed. They generate candidate patterns according to the enumerate strategy and employ depth-first and breadth-first searching methods, respectively.
\end{enumerate} 

\subsection{Mining performance}
	In this subsection, we will evaluate the performance of OPP-Miner from two aspects: time efficiency and scalability.

\subsubsection{Time efficiency}
	To verify the time efficiency of OPP-Miner, we use databases SDB1-SDB8  to carry out the experiments with \textit{minsup} $ = $ 14. Five algorithms are compared in three aspects, the number of OPPs, the number of candidate patterns and running time. The results are shown in Figs.~\ref{fig2:opp}$ -$\ref{fig4:time}.
	
	\begin{figure}[h]
		\centering
		\includegraphics[width=0.35\textwidth]{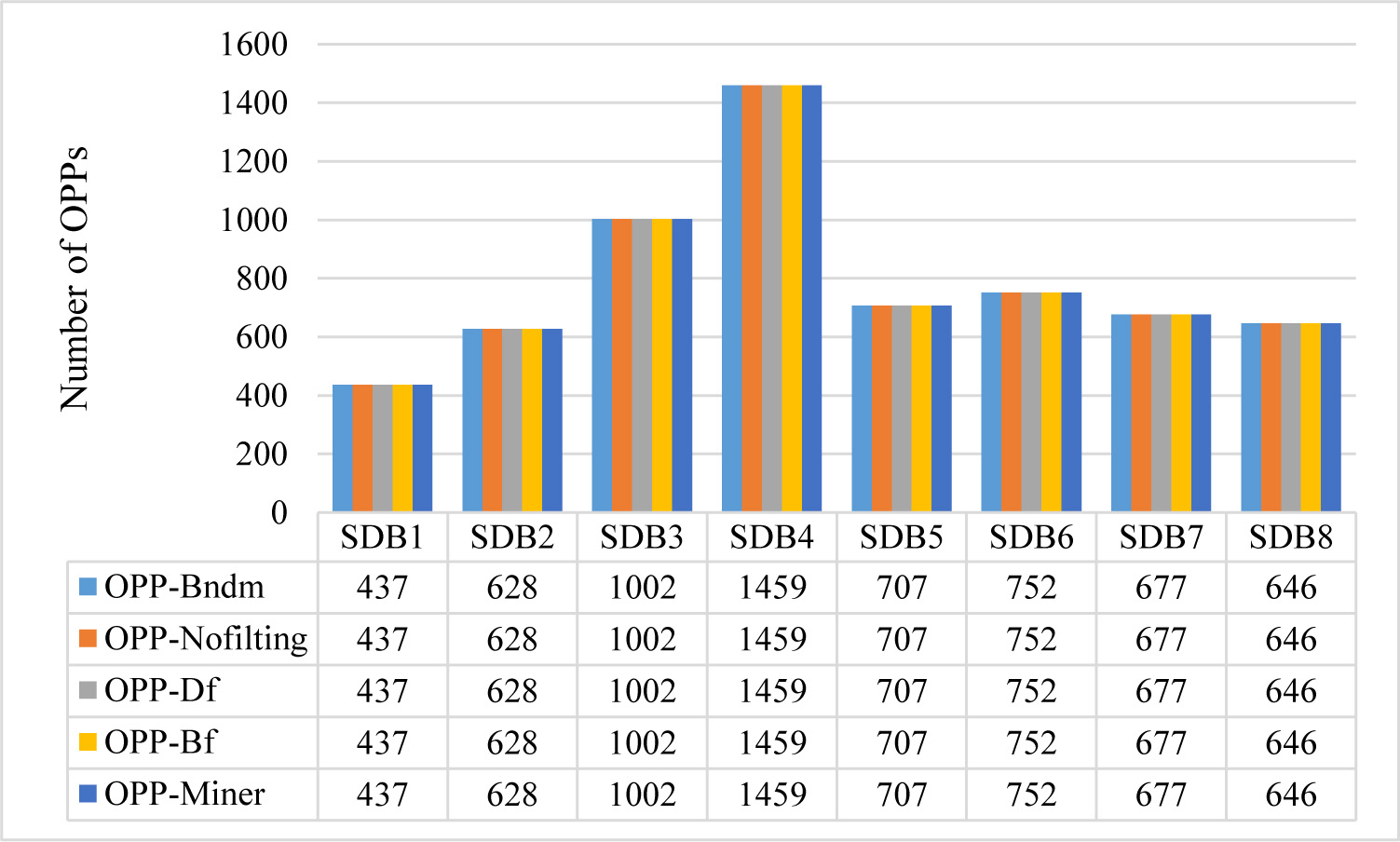}
		\caption{Comparison of  number of OPPs on SDB1 - SDB8}
		\label{fig2:opp}
	\end{figure}

	\begin{figure}[h]
	\centering
	\includegraphics[width=0.35\textwidth]{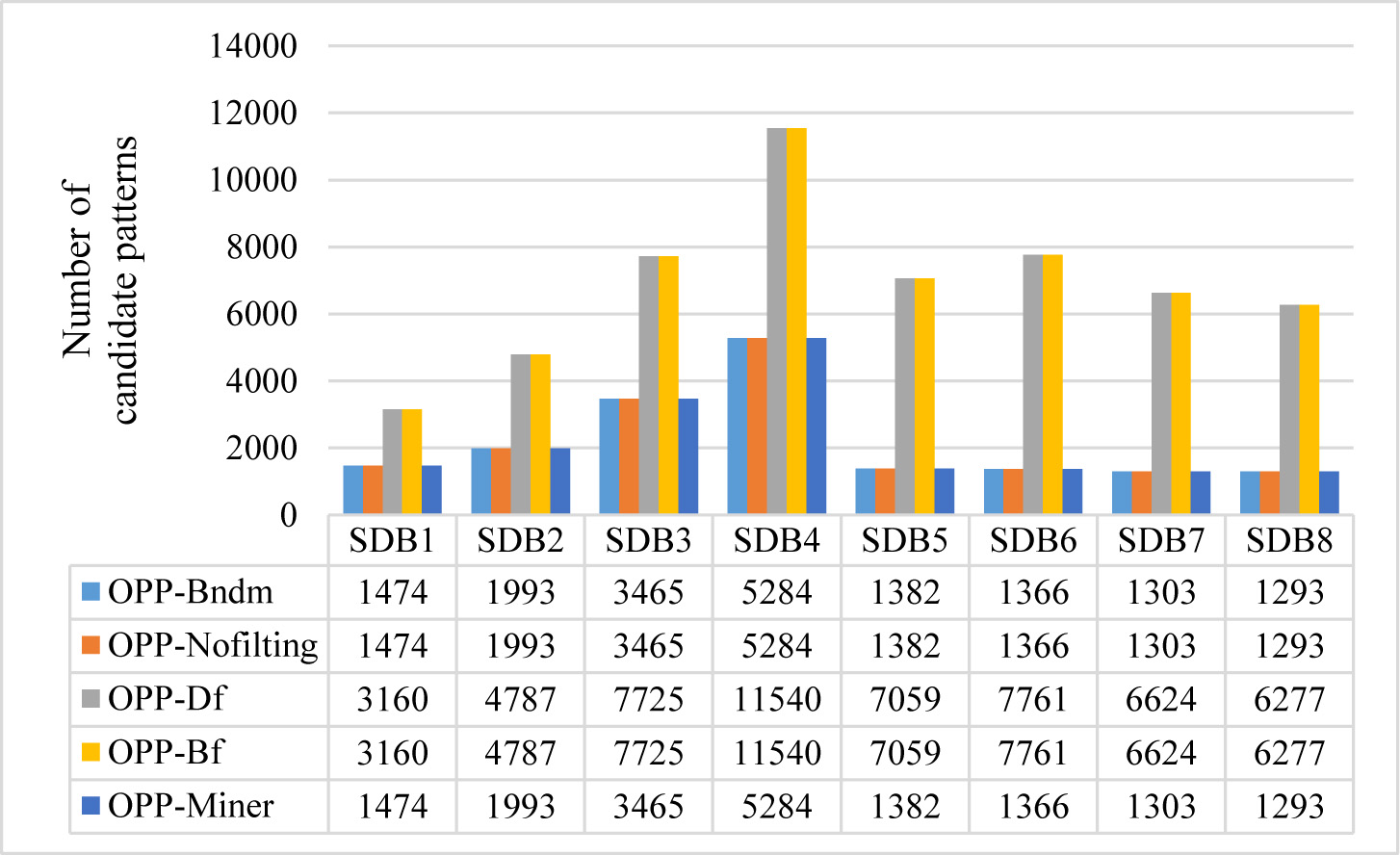}
	\caption{Comparison of  number of candidates on SDB1 - SDB8}
	\label{fig3:candi}
	\end{figure}

	\begin{figure}[h]
	\centering
	\includegraphics[width=0.35\textwidth]{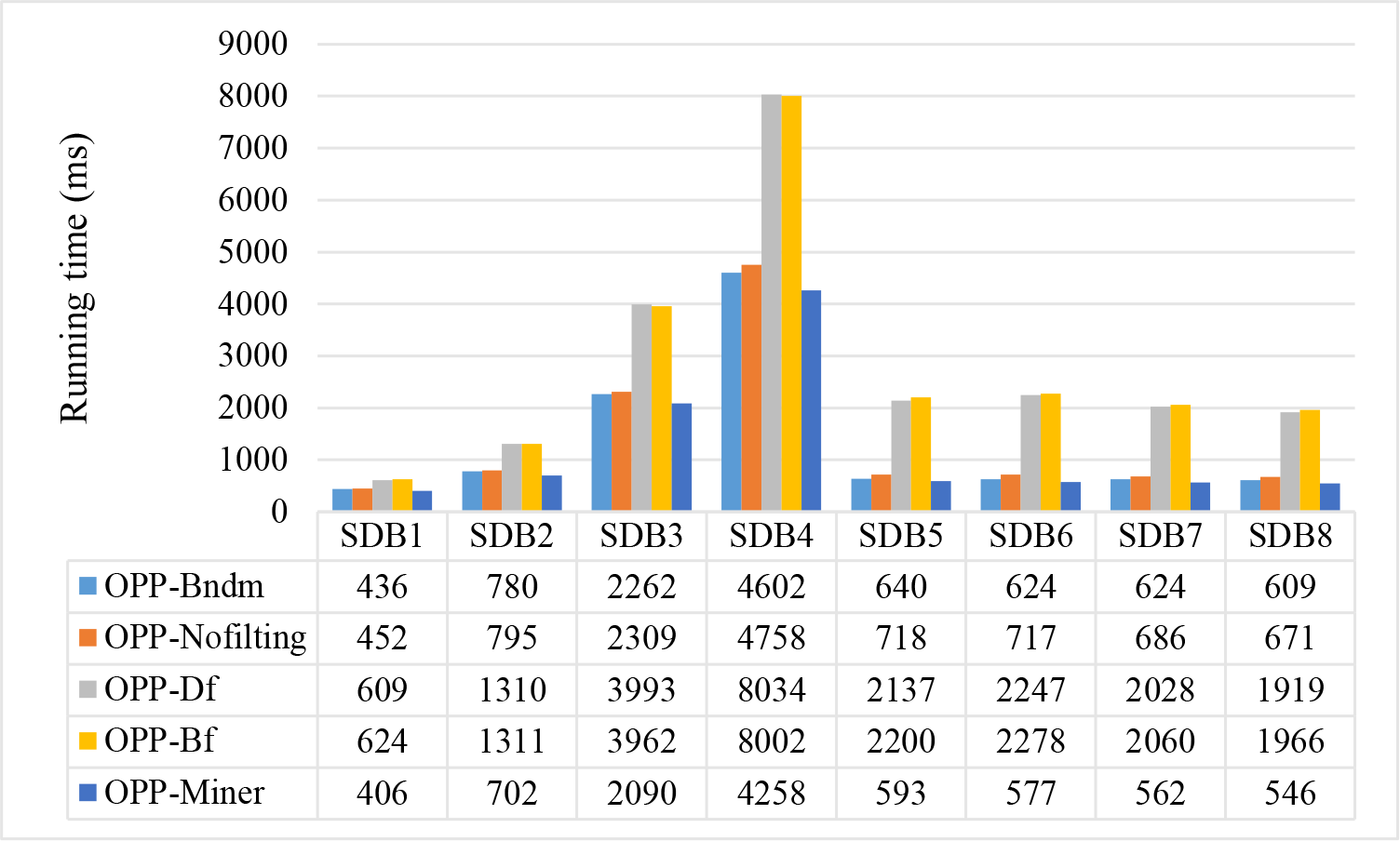}
	\caption{Comparison of running time on SDB1 - SDB8}
	\label{fig4:time}
	\end{figure}

	The results indicate the following observations:
	
\begin{enumerate}
\item The filtration and verification strategies are significantly effective. 

	Firstly, OPP-Miner is faster than OPP-Bndm, which verifies that the $ \rm{SBNDM}_{2} $ algorithm is more efficient than the $ \rm{BNDM} $ algorithm in finding the same sub-sequences as $ \textbf{p}’ $ in $ \textbf{s}’ $. Therefore, we apply the $ \rm{SBNDM}_{2} $ algorithm to OPP-Miner in the filtration strategy. 
	Secondly, OPP-Miner and OPP-Bndm are faster than OPP-Nofilting, which validates the necessity of the filtration strategy. According to Figs.~\ref{fig2:opp}$ -$\ref{fig4:time}, for the number of frequent OPPs and candidate patterns, the results of three algorithms are the same, but OPP-Miner and OPP-Bndm are faster than OPP-Nofilting. The reason is that OPP-Miner and OPP-Bndm add filtration strategy before verification strategy, while OPP-Nofilting does not. This illustrates that, although the filtration strategy will increase time consumption, it can reduce the verifications of redundant sub-sequences, which greatly improves efficiency.
	Hence, OPP-Miner outperforms OPP-Bndm and OPP-Nofilting.

\item Pattern fusion strategy is significantly effective. 

	As shown in Figs.~\ref{fig2:opp} and \ref{fig4:time}, under the same minimum support threshold, the number of frequent OPPs mined by OPP-Miner, OPP-Df and OPP-Bf are the same but OPP-Miner is faster than the other two. The reason is that the number of candidate patterns generated by OPP-Miner is far less than the other two. For example, from Fig.~\ref{fig3:candi}, OPP-Miner generates 1474 candidate patterns in SDB1, while OPP-Df and OPP-Bf generate 3160 candidate patterns. As analysis in Example 6, OPP-Df and OPP-Bf employ the enumeration strategy, which will generate many infrequent candidate patterns, while OPP-Miner employs the pattern fusion strategy which effectively prunes candidate patterns. Hence, OPP-Miner outperforms OPP-Df and OPP-Bf.
\end{enumerate}

	In conclusion, OPP-Miner has better time efficiency than other competitive algorithms.

\subsubsection{Scalability}
	To verify the scalability of OPP-Miner, we carry out experiments under different database sizes. The length of SDB9 is 903000. We intercept the sequence with length of 90, 900, 9000, 90000, and 900000 for experiments, and set \textit{minsup} of 3, 15, 75, 375, and 1875, respectively. The number of OPPs and the running time are shown in Tables~\ref{tab6:different length OPPs} and~\ref{tab7:different length time}.
	
		\begin{table*}
		\centering
		\caption{Comparison of  number of OPPs under different sequence length}
		\label{tab6:different length OPPs}
		\begin{tabular}{cccccccc}
			\toprule   
			 &Length=90&Length=900&Length=9000&Length=90000&Length=900000\\
			\midrule 
			OPP-Bndm &	30	& 34 & 80 & 159 & 212\\
			OPP-Nofilting &	30 &	34 &	80 &	159 & 212\\
			OPP-Df &	30 &	34 &	80 &	159 &	212\\
			OPP-Bf &	30 &	34 &	80 &	159 &	212\\
			OPP-Miner &	30 &	34 &	80 &	159 &	212\\
			\bottomrule
		\end{tabular}
	\end{table*}

		\begin{table*}
	\centering
	\caption{Comparison of  running time under different sequence length (ms)}
	\label{tab7:different length time}
	\begin{tabular}{cccccccc}
		\toprule   
		&Length=90&Length=900&Length=9000&Length=90000&Length=900000\\
		\midrule 
		OPP-Bndm &	31	& 31 & 109 & 2231 & 25678\\
		OPP-Nofilting &	31 &	47 & 109 &	2169 & 24165\\
		OPP-Df &	8 &	17 &	94 &	2090 &	29063\\
		OPP-Bf &	8 &	17 &	94 &	2122 &	29188\\
		OPP-Miner &	\textbf{2} &	\textbf{15} &	\textbf{78} &	\textbf{1950} &	\textbf{22542}\\
		\bottomrule
	\end{tabular}
\end{table*}
	
	As data size increases, OPP-Miner algorithm still has better performance. According to Tables~\ref{tab6:different length OPPs} and~\ref{tab7:different length time}, when the sequence length increases from 90 to 900000, the number of OPPs mined by OPP-Miner is the same as other algorithms, but the time-consuming is the smallest. For example, when the sequence length is 900000, OPP-Miner mines 212 frequent OPPs and takes 22542ms, which is faster than other algorithms. Hence, it can be concluded that the performance of OPP-Miner will not decrease with the increase in data size, i.e., the scalability of the OPP-Miner algorithm is strong.

\subsection{Compression ability}
	In this paper, we propose two mining algorithm, OPP-Miner (mining all frequent OPPs) and MOPP-Miner (mining maximal OPPs), to demonstrate the compression ability of MOPP-Miner, we conduct experiments on SDB1 - SDB8 with \textit{minsup} $ = $ 14, and compare the number of OPPs mined by two algorithms. The results are reported in Table~\ref{tab8:Compression}.
	
	\begin{table*}
		\centering
		\caption{Comparison of  number of OPPs mined by OPP-Miner and MOPP-Miner}
		\label{tab8:Compression}
		\begin{tabular}{cccccccccc}
			\toprule   
			&SDB1&	SDB2&	SDB3&	SDB4&	SDB5&	SDB6&	SDB7&	SDB8\\
			\midrule 
			OPP-Miner&	437&	628&	1002&	1459&	707&	752&	677&	646\\
			MOPP-Miner&	197&	163	& 472&	725&	122&	118&	125&	120\\
			Compress rate	&54.9$ \% $	&74$ \% $	&52.9$ \% $	&50.3$ \% $	&82.7$ \% $	&84.3$ \% $	&81.5$ \% $	&81.4$ \% $\\
			
			\bottomrule
		\end{tabular}
	\end{table*}
	
	Table~\ref{tab8:Compression} shows that MOPP-Miner can compress the patterns effectively. For example, OPP-Miner and MOPP-Miner find 752 frequent patterns and 118 maximal patterns in SDB6, respectively. Thus, the compression rate is (752$ - $118) / 752 $ = $ 84.3$ \% $. The reason is that OPP-Miner mines the complete set of the frequent patterns which contains redundant patterns, while MOPP-Miner mines a subset of the maximal patterns according to the maximal checking strategy. Therefore, the MOPP-Miner algorithm achieves the compression of the results, which will simplify the data understanding process greatly.

\subsection{Mining ability}
	Finding the critical trends is a typical task of time series pattern discovery. The OPP is used exactly to represent the trend of a time series based on order relation, which means that an OPP represents a kind of trend. Hence, if an OPP is reproduced in the original time series, the order relations of all occurrences should be exactly the same and the actual trends of all occurrences should be similar. Therefore, to verify this assumption, we conduct experiments on SDB10 with \textit{minsup} $ = $ 3 and select four OPPs, which are (3, 1, 2, 5, 4), (1, 2, 5, 3, 4), (5, 3, 4, 1, 2) and (4, 5, 2, 3, 1). We reproduce their occurrences in the original time series. The results are shown in Fig. \ref{FIG:5 similar sub-sequences}.
	
	\begin{figure*}[ht]
		\centering
		\includegraphics[width=0.7\textwidth]{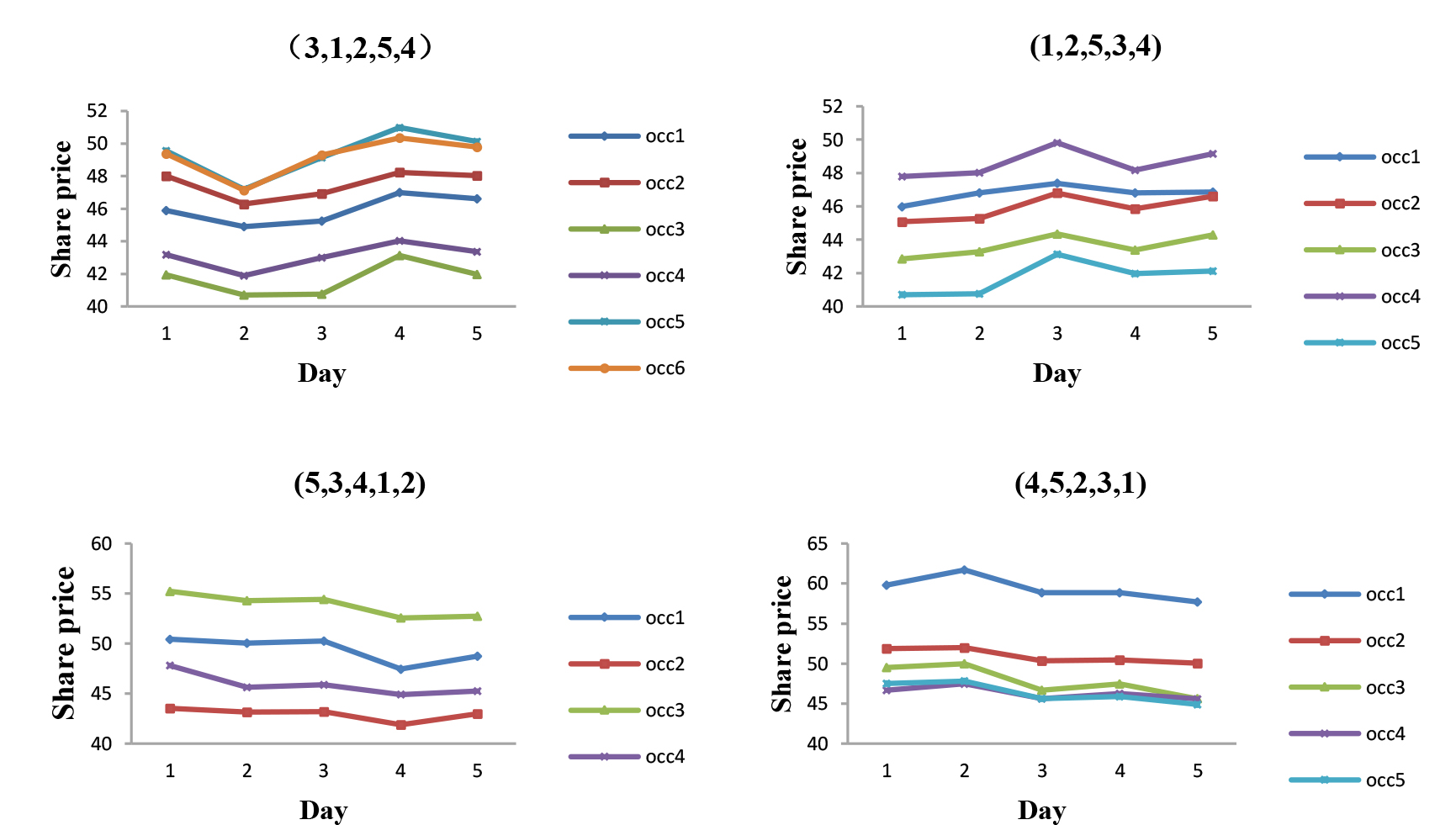}
		\caption{Example of  trends of four OPPs on SDB10.}
		\label{FIG:5 similar sub-sequences}
	\end{figure*}
	
	From Fig. \ref{FIG:5 similar sub-sequences}, it is clear that the order relations of all occurrences of each OPP are exactly the same, which proves the effectiveness of the OPP algorithm. Meanwhile, precisely because of the same order relations, all trends of the occurrences of each OPP in the original time series are very similar. For example, for OPP (3, 1, 2, 5, 4), the trends of $ occ_{5} $ and $ occ_{6} $ are almost coincident, and although other occurrences have different starting points, their overall shapes are very similar. Hence, OPP-Miner can find similar sub-sequences. 

\subsection{OPP Mining Case Study}
\subsubsection{COVID-19 Critical Trend Identification}
	Since early 2020, COVID-19 virus has been spreading around the world. As the epidemic intensifies worldwide, various data have been aggregated to evaluate the epidemic progression, such as the total confirmed cases, the daily confirmed new cases and the total deaths. Finding effective tools/algorithms to analyze such data is crucial to understand the spread pattern of the disease. 

	OPP mining can serve as an analytical method to identify critical trends of epidemics by mining frequent OPPs. In this subsection, we select the daily COVID-19 new cases data in China, USA, Brazil, and Iran from January 22, 2020 to July 20, 2020. To prevent data skew, we use a 5-day moving average, \textit{i.e.}, averaging the values within a five day window (including current day and two days before and after). OPP-Miner is employed to mine frequent OPPs with \textit{minsup} $ = $ 20 and the results are reported in Fig. \ref{FIG:6 COVID-19}. For ease of understanding, the OPPs representing upward trends are colored in yellow, and OPPs representing downward trends are colored in blue in Fig. \ref{FIG:6 COVID-19},
	
		\begin{figure*}[ht]
		\centering
		\includegraphics[width=0.45\textwidth]{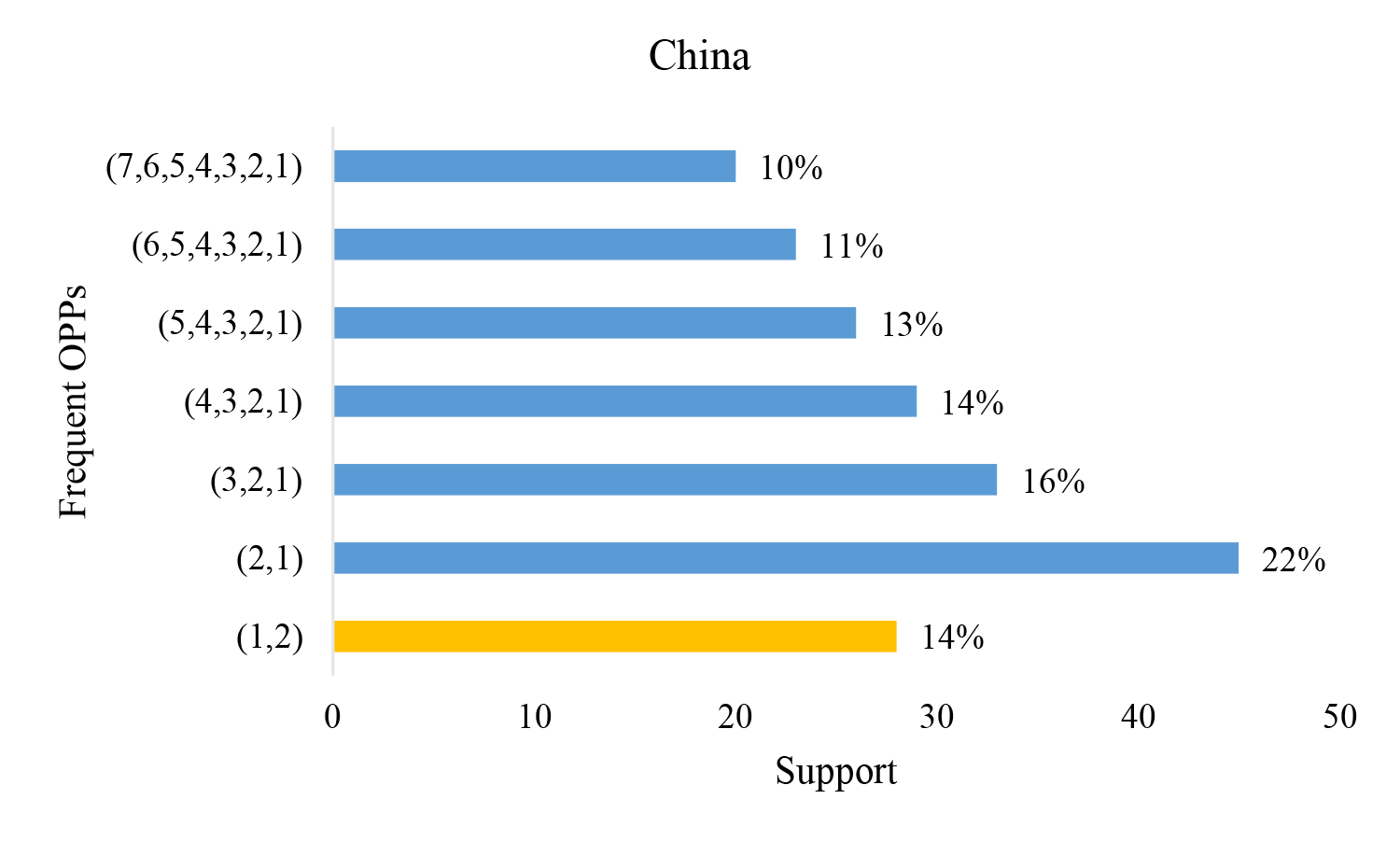}
		\includegraphics[width=0.45\textwidth]{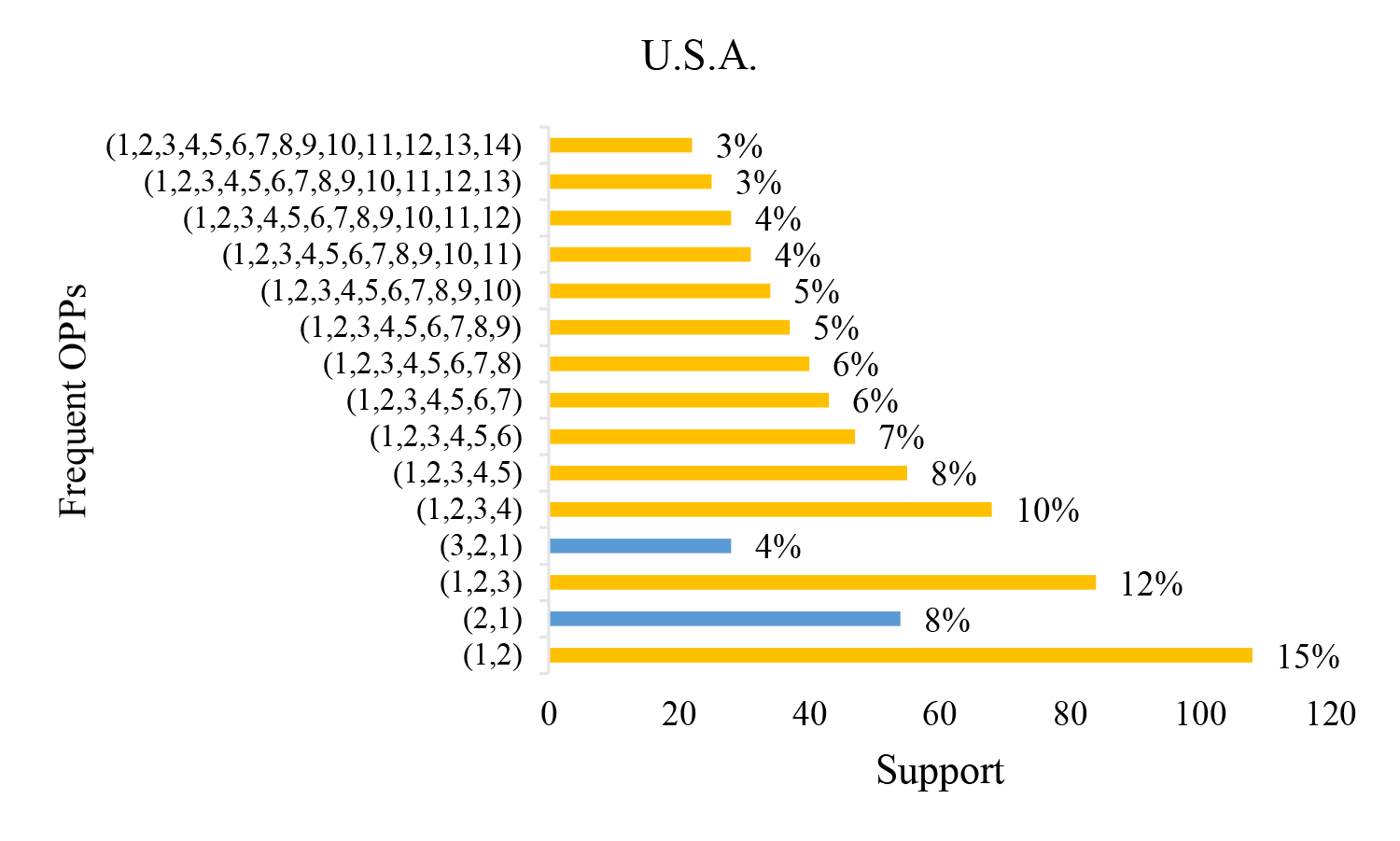}
		\includegraphics[width=0.45\textwidth]{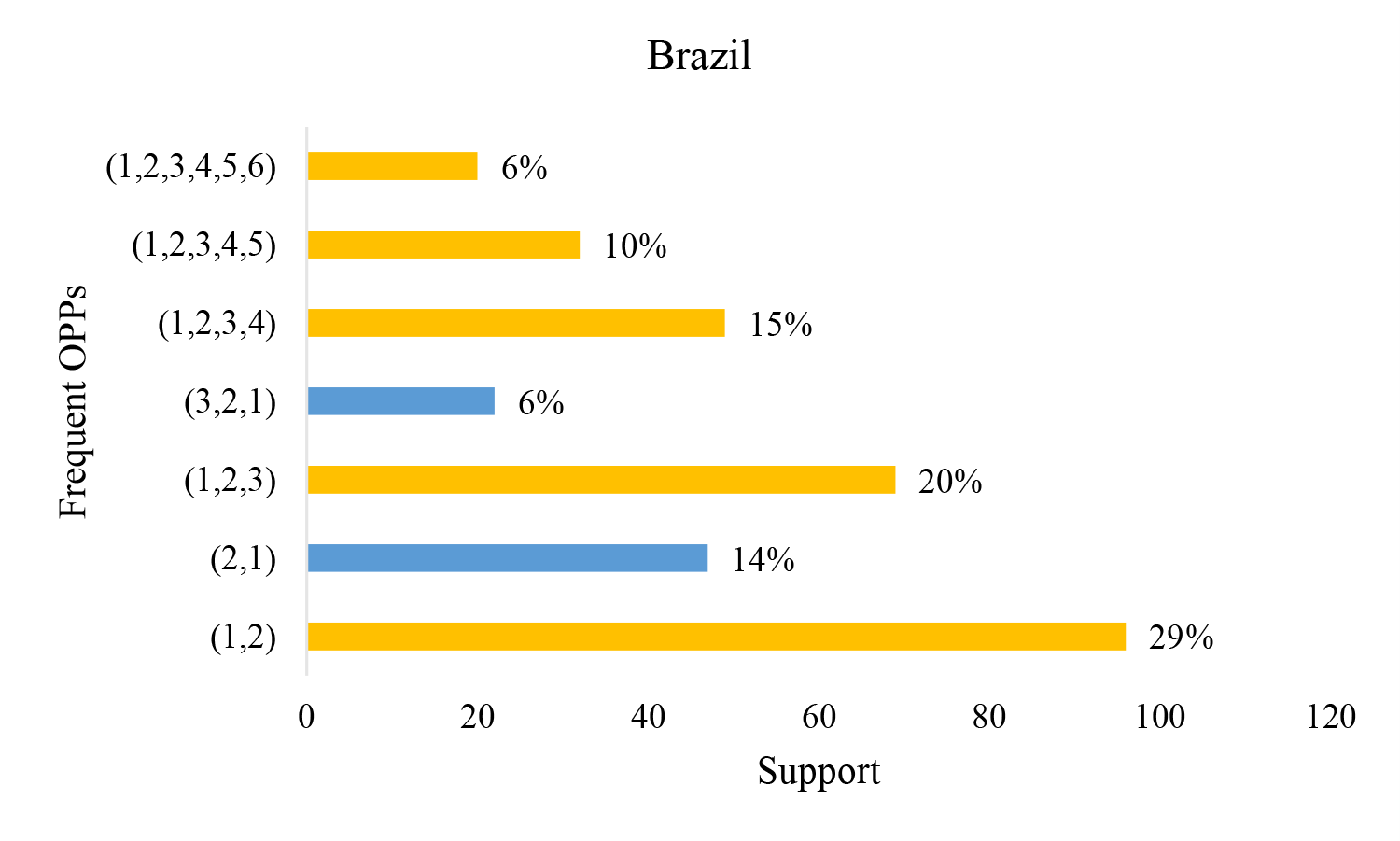}
		\includegraphics[width=0.45\textwidth]{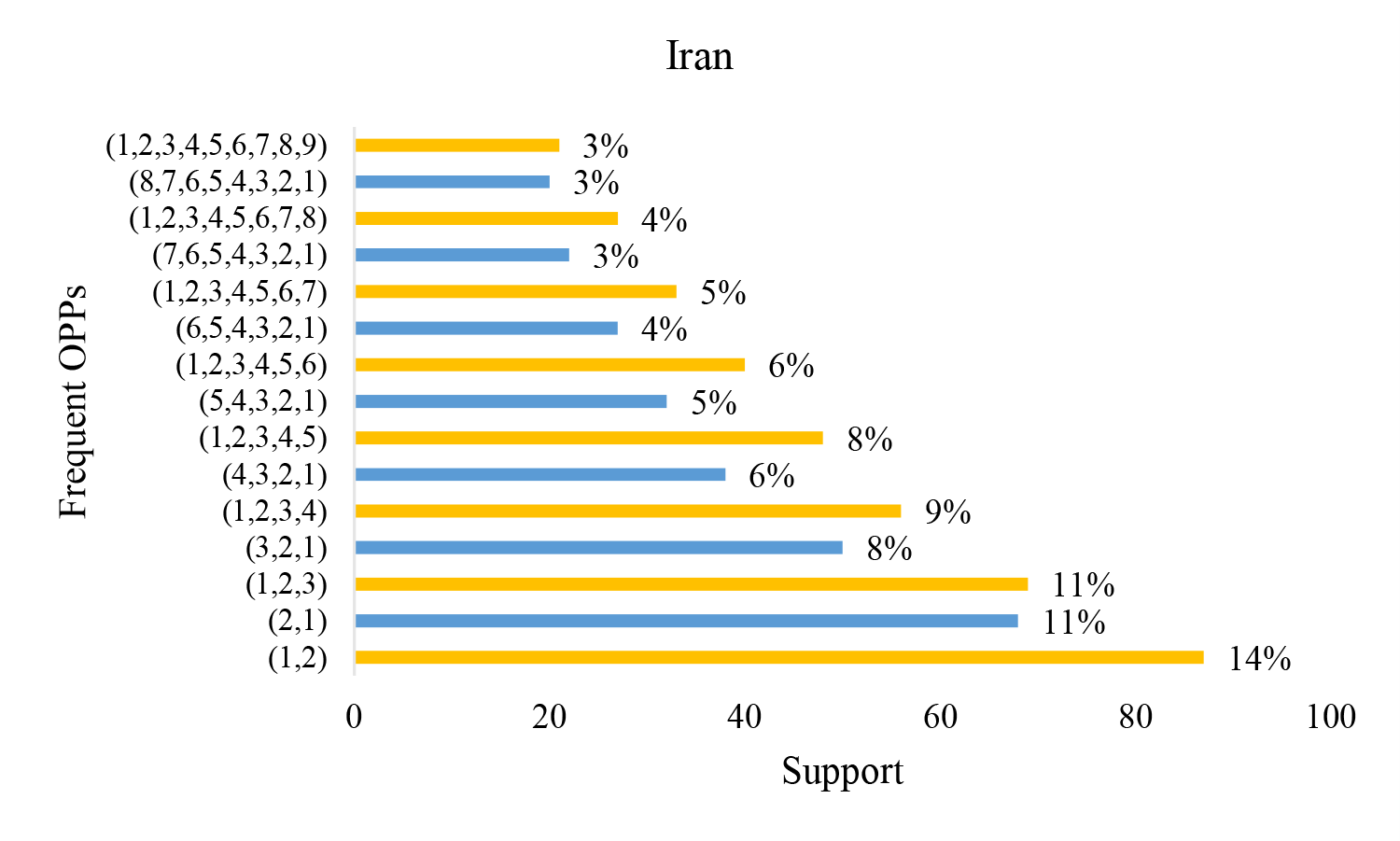}
		\caption{Frequent OPPs discovered from COVID-19 infection time series of four countries (OPPs representing upward trends are colored in yellow and OPPs representing downward trends are colored in blue)}
		\label{FIG:6 COVID-19}
	\end{figure*}

	The results from the COVID-19 data show the following observations. For China, the epidemic was gradually under control and recovering. According to Fig. \ref{FIG:6 COVID-19}, 86 percent of OPPs presented a downward trend and only 14 percent of OPPs were in an upward trend, which was a sign that the epidemic was brought under control. For USA and Brazil, the epidemic was in a severe outbreak. As reported in Fig. \ref{FIG:6 COVID-19}, the majority of OPPs were in an upward trend, indicating the epidemic in these countries was very serious. For Iran, the proportion of OPPs in a downward trend and in an upward trend was close, indicating the daily new cases were decreasing and the confirmed cases growth was slowing down. Therefore, the OPP mining method can help people analyse the epidemic situation by identifying critical trends.

\subsubsection{Clustering}
	To further demonstrate the utility of our algorithm, we also carry out the clustering experiment. Three datasets are selected as the raw data, which are SDB11, SDB12 and SDB13. For each dataset, we process it as follows.
	
	\begin{enumerate} 
	\item MOPP-Miner is employed to mine the maximal OPPs.
	
	\item Record the maximal OPPs and their supports as the mining data.
	
	\item K-means~\cite{Capo2020efficient} method is adopted to cluster the raw data and the mining data, respectively.
	
	\item Two metrics are used to evaluate the clustering performance, which are Normalized Mutual Information (\textit{NMI})~\cite{Danon2005Comparing} and Homogeneity (\textit{h})~\cite{Rosenberg2007V-Measure}. They can be calculated by Equations (3) and (4), respectively.
	\end{enumerate} 

	\begin{equation}
	NMI (X, Y) =\frac
	{\sum_{i = 1} ^{\lvert X \rvert}\sum_{j = 1} ^{\lvert Y \rvert}P(i,j)log(\frac{P(i,j)}{P(i)P(j)})}
	{\sqrt{\sum_{i = 1} ^{\lvert X \rvert}P(i)logP(i)\sum_{j = 1} ^{\lvert Y \rvert}P(j)logP(j)}}
	\end{equation}
	
	\begin{equation}
	h (X, Y) =1 - \frac
	{-\sum_{i = 1} ^{\lvert X \rvert}\sum_{j = 1} ^{\lvert Y \rvert}P(i,j)logP(i\lvert j)}
	{-\sum_{i = 1} ^{\lvert X \rvert}P(i)logP(i)}
	\end{equation}
		
	The experiments are conducted with \textit{minsup} $ = $ 6, 5, 5, and K $ = $ 4, 4, 5, respectively. The results are shown in Table \ref{tab9:clustering}.
	
	\begin{table}[h]
		\centering
		\caption{Comparison of clustering performance}
		\label{tab9:clustering}
		\begin{tabular}{ccccccc}
			\toprule   
				&Data	&Dimensionality	&\textit{NMI}	&\textit{h}\\
			\midrule 
			SDB11& Raw data&	181&	0.52&	0.46\\
			&Mining data	&\textbf{15}	&\textbf{0.60}	&\textbf{0.61}\\
			SDB12&Raw data	&577	&0.46	&0.45\\
			&Mining data	&\textbf{8}	&\textbf{0.59}	&\textbf{0.58}\\
			SDB13&Raw data	&470	&0.48	&0.45\\
			&Mining data	&\textbf{10}	&\textbf{0.63}	&\textbf{0.62}\\
			
			\bottomrule
		\end{tabular}
	\end{table}
	
	As shown in Table \ref{tab9:clustering}, the dimensionality of the raw data of SDB11 is 181. But through the maximal OPP mining, the dimensionality is reduced to 15. This has many advantages. Firstly, the reduction of data dimensionality will simplify the calculation. More importantly, using the maximal OPPs can improve the clustering performance. For example, \textit{NMI} and \textit{h} all reflect the similarity between the clustering results and the actual values. The greater the \textit{NMI} and \textit{h} are, the more similarity the clustering results and the actual values. For SDB13, the \textit{NMI} and \textit{h} of the raw data clustering result are 0.48 and 0.45, respectively, while the mining data are 0.63 and 0.62, respectively. Both of the two evaluation metrics show that using the maximal OPPs can improve the clustering performance. The reason is that the raw data generally contains many redundant information which will affect the clustering performance, while the maximal OPP mining implements the feature extraction which is useful for clustering.

\section{Conclusion}
\label{section:Conclusion}
	In this paper, we study a novel order-preserving pattern (OPP) mining problem in time series, where an OPP represents the trend of a time series based on the order relations of the time series data. We argue that OPP pattern has unique advantages, compared to other times series patterns, mainly because an OPP can summarize trend in the data, and does not require explicit conversion of time series as nominal/symbolic for pattern mining. However, mining OPPs is difficult because of the large pattern space. To tackle the challenge, we propose an OPP-Miner algorithm, which adopts the filtration and verification strategies to calculate support and the pattern fusion strategy to generate candidate patterns. In the process of support calculation, although the verification strategy can directly find the occurrences, the filtration strategy can avoid invalid verifications on redundant sub-sequences. In the process of candidate patterns generation, the enumeration strategy can generate candidate patterns, but it generates too many candidate patterns. To improve the efficiency, we propose the pattern fusion strategy which implements pruning by fusing frequent patterns satisfying the fusion conditions. Moreover, we also develop an MOPP-Miner algorithm to find maximal OPPs, which compresses the result set effectively. We evaluate the performance of OPP-Miner using real-life time series databases, such as stock, temperature and the daily COVID-19 new cases data. Experimental results demonstrate that our algorithms outperform competitive algorithms. More importantly, the maximal OPP mining implements the feature extraction, which can improve the clustering performance.


\begin{thebibliography}{99}
	

\bibitem{Fournier-Viger2017a}  
P. Fournier-Viger, J. C. W. Lin, R. U. Kiran, Y. S. Koh, and R. Thomas, ``A survey of sequential pattern mining,'' \textit{Data Sci. Pattern Recognit}., vol. 1, no. 1, pp. 54$ - $77, 2017.

\bibitem{wu2021tcyb}
Y.  Wu, Y. Wang, Y. Li, X. Zhu, and X Wu,  ``Top-k self-adaptive contrast sequential pattern mining,'' \textit{IEEE Transactions on Cybernetics}., DOI: 10.1109/TCYB.2021.3082114, 2021.

\bibitem{qiang2020tkde}  
J. Qiang, Z. Qian, Y. Li, Y. Yuan, and X. Wu, ``Short text topic modeling techniques, applications, and performance: A survey,'' \textit{IEEE Transactions on Knowledge and Data Engineering}., DOI:10.1109/TKDE.2020.2992485, 2020.

\bibitem{Ghosh2017Septic} 
S. Ghosh, J. Li, L. Cao, and K. Ramamohanarao, ``Septic shock prediction for ICU patients via coupled HMM walking on sequential contrast patterns,'' \textit{J. Biomed. Inform}., vol. 66, pp. 19$ - $31, 2017.



\bibitem{fournier2020event} 
P. Fournier-Viger, J. Li, J. C. W. Lin, T. Truong, and R. U. Kiran, ``Mining cost-effective patterns in event logs,'' \textit{Knowl. Based Syst}., vol. 191, pp. 105241, 2020.

\bibitem{Nishimuraa2018latent} 
N. Nishimuraa, N. Sukegawab, Y. Takanoc, and J. Iwanagad, ``A latent-class model for estimating product-choice probabilities from clickstream data,'' \textit{Inf. Sci}., vol. 429, pp. 406$ - $420, 2018.





\bibitem{Dong2020e-RNSP}
X. Dong, Y. Gong, and L. Cao, ``e-RNSP: An efficient method for mining repetition negative sequential patterns,'' \textit{IEEE Trans. Cybern}., vol. 50, no. 5, pp. 2084$ - $2096, 2020.


\bibitem{Dong2019top-k}
X. Dong, P. Qiu, J. Lü, L. Cao, and T. Xu, ``Mining top-k useful negative sequential patterns via learning,'' \textit{IEEE Trans. Neural Netw. Learn. Syst}., vol. 30, no. 9, pp. 2764$ - $2778, 2019.

\bibitem{Truong2020EHAUSM} 
T. Truong, H. Duong, B. Le, and P. Fournier-Viger, ``EHAUSM: An efficient algorithm for high average utility sequence mining,'' \textit{Inf. Sci}., vol. 515, pp. 302$ - $323, 2020.

\bibitem{Gan2020HUOPM} 
W. Gan, J. C. Lin, P. Fournier-Viger, H. Chao, and P. S. Yu, ``HUOPM: High-utility occupancy pattern mining,'' \textit{IEEE Trans. Cybern}.,
vol. 50, no. 3, pp. 1195$ - $1208, 2020.

\bibitem{fournier2019} 
P. Fournier-Viger, Y. Zhang, J. C. W. Lin, H. Fujita, and Y. S. Koh, ``Mining local and peak high utility itemsets,'' \textit{Inf. Sci}., vol. 481, pp. 344$ - $367, 2019.

\bibitem{Wu2018NOSEP} 
Y. Wu, Y. Tong, X. Zhu, and X. Wu, ``NOSEP: Nonoverlapping sequence pattern mining with gap constraints,'' \textit{IEEE Trans. Cybern}., vol. 48, no. 10, pp. 2809$ - $2822, 2018.

\bibitem{Truong2019Efficient} 
T. Truong, H. Duong, B. Le and P. Fournier-Viger, ``Efficient vertical mining of high average-utility itemsets based on novel upper-bounds,'' \textit{IEEE Trans. Knowl. Data Eng}., vol. 31, no. 2, pp. 301$ - $314, 2019.

\bibitem{song2021kais}
W. Song, L. Liu, and C. Huang,  ``Generalized maximal utility for mining high average-utility itemsets,''  \textit{Knowledge and Information Systems}., vol. 63, pp. 2947–2967, 2021.

\bibitem{Wu2014Mining} 
Y. Wu, L. Wang, J. Ren, W. Ding, and X. Wu, ``Mining sequential patterns with periodic wildcard gaps,'' \textit{Appl. Intell}., vol. 41, no. 1, pp. 99$ - $116, 2014.

\bibitem{Keogh2001online} 
E. Keogh, S. Chu, D. Hart, and M. Pazzani, ``An online algorithm for segmenting time series,'' in \textit{Proc. IEEE Int. Conf. Data Min.}, 2001, pp. 289$ - $296.

\bibitem{Lin2007Experiencing} 
J. Lin, E. Keogh, L. Wei, and S. Lonardi, ``Experiencing SAX: a novel symbolic representation of time series,'' \textit{Data Min. Knowl. Discov}.,
vol. 15, no. 2, pp. 107$ - $144, 2007.


\bibitem{li2021apind}
Y. Li, L. Yu, J. Liu, L. Guo, Y. Wu, and X. Wu, ``NetDPO: (delta, gamma)-approximate pattern matching with gap constraints under one-off condition,'' \textit{Applied Intelligence}., DOI: 10.1007/s10489-021-03000-2, 2021.

\bibitem{wu2022ida}
Y. Wu, B. Jian, Y. Li, H. Jiang, and X. Wu, ``NetNDP: Nonoverlapping (delta, gamma)-approximate pattern matching,'' \textit{Intelligent Data Analysis}. DOI: 10.3233/IDA-216325. 2022.



\bibitem {wu2021tmis}
Y. Wu, X. Wang, Y. Li, L. Guo, Z. Li, J. Zhang, and X. Wu, ``OWSP-Miner: Self-adaptive one-off weak-gap strong pattern mining,'' \textit{ACM Transactions on Management Information Systems}., DOI: 10.1145/3476247. 2022.

\bibitem{wu2021insweak}
Y. Wu, Z. Yuan, Y. Li, L. Guo, P. Fournier-Viger, and X. Wu, ``NWP-Miner: Nonoverlapping weak-gap sequential pattern mining,'' \textit{Information Sciences}, 588,  pp. 124-141,  2022.


\bibitem{Kim2014Order} 
J. Kim, P. Eades, R. Fleischer, S. Hong, C. S. Iliopoulos, K. Park, S. J. Puglisi, and T. Tokuyama, ``Order-preserving matching,'' \textit{Theor. Comput. Sci}., vol. 525, pp. 68$ - $79, 2014.

\bibitem{Wu2020NetNCSP} 
Y. Wu, C. Zhu, Y. Li, L. Guo, and X. Wu, ``NetNCSP: Nonoverlapping closed sequential pattern mining,'' \textit{Knowl.	Based Syst}., 2020, doi: 10.1016/j.knosys.2020.105812.

\bibitem{Yun2014maximal} 
U. Yun, G. Lee, and K. H. Ryu, ``Mining maximal frequent patterns by considering weight conditions over data streams,'' \textit{Knowl. Based Syst}.,
vol. 55, pp. 49$ - $65, 2014.

\bibitem{li2021apinmax}
Y. Li, S. Zhang, L. Guo, J. Liu, Y. Wu, and X. Wu, ``NetNMSP: Nonoverlapping maximal sequential pattern mining,'' \textit{Applied Intelligence}., 2021, DOI: 10.1007/s10489-021-02912-3

\bibitem{Huang2019Mining} 
J. Huang, B. P. Jaysawal, K. Chen, and Y. Wu, ``Mining frequent and top-k high utility time interval-based events with duration patterns,'' \textit{Knowl. Inf. Syst}., vol. 61, no. 3, pp. 1331$ - $1359, 2019.

\bibitem{Dam2016efficient} 
T. Dam, K. Li, P. Fournier-Viger, and Q. Duong, ``An efficient algorithm for mining top-rank-k frequent patterns,'' \textit{Appl. Intell}., vol. 45, no. 1, pp. 96$ - $111, 2016.

\bibitem{Min2020Frequent} 
F. Min, Z. Zhang, W. Zhai, and R. Shen, ``Frequent pattern discovery with tri-partition alphabets,'' \textit{Inf. Sci}., vol. 507, pp. 715$ - $732, 2020.

\bibitem{wu2022tkdd}
Y. Wu, L. Luo,  Y. Li, L. Guo, P. Fournier-Viger, X. Zhu, and  X. Wu, ``NTP-Miner: Nonoverlapping three-way sequential pattern mining,'' \textit{ACM Transactions on Knowledge Discovery from Data},  2022, 16(3): 51.  DOI: 10.1145/3480245. 

\bibitem{gan2021tcyb}
W. Gan, J. C. W. Lin, J. Zhang, P. Fournier-Viger, H. C. Chao, and P. S. Yu. ``Fast Utility Mining on Sequence Data,'' \textit {IEEE Trans. Cybern}., vol. 51, no. 2, pp. 487-500, 2021.



\bibitem{Wang2016Efficient} 
T. Wang, L. Duan, G. Dong, and Z. Bao, ``Efficient mining of outlying sequence patterns for analyzing outlierness of sequence data,'' \textit{ACM Trans. Knowl. Discov. Data}., vol. 14, no. 5, pp. 62, 2020.


\bibitem{Piri2018Development} 
S. Piri, D. Delen, T. Liu, and W. Paiva, ``Development of a new metric to identify rare patterns in association analysis: The case of analyzing diabetes complications,'' \textit{Expert Syst. Appl}., vol. 94, pp. 112$ - $125, 2018.

\bibitem{Wang2018co} 
L. Wang, X. Bao, and L. Zhou, ``Redundancy reduction for prevalent co-location patterns,'' \textit {IEEE Trans. Knowl. Data Eng}., vol. 30, no.1, pp. 142-155, 2018.


\bibitem{Lin2020Utilitye}
J. C.  W. Lin, T. Li, M. Pirouz, J. Zhang, and P. Fournier-Viger, ``High average-utility sequential pattern mining based on uncertain databases,'' \textit{Knowl. Inf. Syst}., vol. 62, no. 3, pp. 1199$ – $1228, 2020.

\bibitem{Xu2015short} 
Q. Xu, D. He, N. Zhang, C. Kang, Q. Xia, J. Bai, and J. Huang, ``A short-term wind power forecasting approach with adjustment of numerical weather prediction input by data mining,'' \textit{IEEE Trans. Sustain. Energy}., vol. 6, no. 4, pp. 1283$ - $1291, 2015.

\bibitem{Samiee2015Epileptic} 
K. Samiee, P. Kovács, and M. Gabbouj, ``Epileptic seizure classification of EEG time-series using rational discrete short-time fourier transform,'' \textit{IEEE Trans. Biomed. Eng}., vol. 62, no. 2, pp. 541$ - $552, 2015.

\bibitem{Li2020multimodal} 
Q. Li, J. Tan, J. Wang, and H. Chen, ``A multimodal event-driven LSTM model for stock prediction using online news,'' \textit{IEEE Trans. Knowl. Data Eng}., 2020, doi: 10.1109/TKDE.2020.2968894.




\bibitem{Tsai2015location} 
C. Tsai, and B. Lai, ``A location-item-time sequential pattern mining algorithm for route recommendation,'' \textit{Knowl.
	Based Syst}., vol. 73, pp. 97$ - $110, 2015.



\bibitem{wugeng2021kbs}
Y. Wu, M. Geng, Y. Li, L. Guo, Z. Li, P. Fournier-Viger, X. Zhu, and X. Wu, ``HANP-Miner: High average utility nonoverlapping sequential pattern mining,'' \textit{Knowledge-Based Systems}., 229, 107361, 2021.

\bibitem{wulei2021eswa}
Y. Wu, R. Lei, Y. Li, L. Guo, and X. Wu, ``HAOP-Miner: Self-adaptive high-average utility one-off sequential pattern mining,'' \textit{Expert Systems With Applications}., 184, 115449, 2021. 


\bibitem{Guo2010improved} 
C. Guo, H. Li, and D. Pan, ``An improved piecewise aggregate approximation based on statistical features for time series mining,'' in \textit{Proc. Int. Conf. Knowl. Sci. Eng. Man}., 2010, pp. 234$ - $244.


\bibitem{Lee2003Dimensionality} 
S. Lee, D. Kwon, and S. Lee, ``Dimensionality reduction for indexing time series based on the minimum distance,'' \textit{J. Inf. Sci. Eng}., vol. 19, pp. 697$ - $711, 2003.

\bibitem{Lin2002Finding} 
J. Lin, E. Keogh, S. Lonardi, and P. Patel, ``Finding motifs in time series,'' in \textit{Proc. Workshop on Temporal Data Mining}, 2002, pp. 53$ - $68.

\bibitem{Keogh2005HOT} 
E. Keogh, J. Lin, and A. Fu, ``HOT SAX: Efficiently finding the most unusual time series subsequence,'' in \textit{Proc. IEEE Int. Conf. Data Min}., 2005, pp. 226$ - $233.


\bibitem{Tan2016Discovering} 
C. Tan, F. Min, M. Wang, H. Zhang, and Z. Zhang, ``Discovering patterns with weak-wildcard gaps,'' \textit{IEEE Access}, vol. 4, pp. 4922$ - $4932, 2016.

\bibitem{Cho2013Fast} 
S. Cho, J. C. Na, K. Park, and J. S. Sim, ``Fast order-preserving pattern matching,'' \textit{Combinatorial Optimization and Applications}., vol. 8287, pp. 295$ - $305, 2013.

\bibitem{Crochemore2013Order} 
M. Crochemore, C. S. Iliopoulos, T. Kociumaka, M. Kubica, A. Langiu, S. P. Pissis, J. Radoszewski, W. Rytter, and T. Waleń, ``Order-preserving incomplete suffix trees and order-preserving indexes,'' in \textit{Proc. String Processing and Information Retrieval}, 2013, 84$ - $95.

\bibitem{Chhabra2016filtration} 
T. Chhabra, and J. Tarhio, ``A filtration method for order-preserving matching,'' \textit{Inf. Process. Lett}., vol. 116, no. 2, pp. 71$ - $74, 2016.

\bibitem{Chhabra2016Engineering} 
T. Chhabra, S. Faro, M. O. Külekci, and J. Tarhio, ``Engineering order-preserving pattern matching with SIMD parallelism,'' \textit{Software Pract. Exper.}, 2016, doi: 10.1002/spe.2433.

\bibitem{Pawel2016Order} 
G. Paweł, and U. Przemysław, ``Order-preserving pattern matching with \textit{k} mismatches,'' \textit{Theor. Comput. Sci}., vol. 638, pp. 136$ - $144, 2016.


\bibitem{Juan2018New} 
M. Juan, N. Rafael, P. Yoan, and H. Germán, ``New algorithms for delta, gamma -order preserving matching,'' \textit{Ingeniería}, vol. 23, no. 2, pp. 190$ - $202, 2018.
	
\bibitem{Wu2020NetDAP} 
Y. Wu, J. Fan, Y. Li, L. Guo, and X. Wu, ``NetDAP: (delta, gamma )-Approximate pattern matching with length constraints,'' \textit{Appl. Intell}., vol. 50, no. 11, pp. 4094$ - $4116, 2020.

\bibitem{Wu2017Strict} 
Y. Wu, C. Shen, H. Jiang, and X. Wu, ``Strict pattern matching under non-overlapping condition,'' \textit{Sci. China Inf. Sci}., vol. 60, no. 1, pp. 1$ - $16, 2017.

\bibitem{Durian2010Improving} 
B. Ďurian, J. Holub, H. Peltola, and J. Tarhio, ``Improving practical exact string matching,'' \textit{Inf. Process. Lett}., vol. 110, pp. 148$ - $152, 2010.

\bibitem{Zhang2017Cautionary} 
S. Zhang, B. Guo, A. Dong, J. He, Z. Xu, and S. Chen, ``Cautionary tales on air-quality improvement in Beijing,'' \textit{Math. Phys. Eng. Sci}., vol. 473, no. 2205, pp. 20170457, 2017.


\bibitem{Murray2015Data}
D. Murray, ``A data management platform for personalised real-time energy feedback,'' in \textit{Proc. Int. Conf. Energy Efficiency Domestic Appl. Lighting}, 2015, pp. 1293$ - $1307.

\bibitem{Al-Jowder2002Detection} 
O. Al-Jowder, E. K. Kemsley, and R.H Reginald, ``Detection of adulteration in cooked meat products by mid-infrared spectroscopy,'' \textit{J. Agric. Food Chem}.,  vol. 50, no. 6, pp. 1325$ - $1329, 2002.

\bibitem{Capo2020efficient} 
M. Capo, A. Perez, and J. A. A. Lozano, ``An efficient Split-Merge re-start for the K-means algorithm,'' \textit{IEEE Trans. Knowl. Data Eng.}, doi: 10.1109/TKDE.2020.3002926.


\bibitem{Danon2005Comparing} 
L. Danon, A. Diaz-Guilera, J. Duch, and A. Arenas, ``Comparing community structure identification,'' \textit{J. Stat. Mech.-Theory Exp}., vol. 2005, no. 09, pp. P09008, 2005.

\bibitem{Rosenberg2007V-Measure} 
A. Rosenberg and J. Hirschberg, ``V-Measure: A conditional entropy-based external cluster evaluation measure,'' in \textit{Proc. Conf. Empirical Methods
Natural Lang. Process}., 2007, pp. 410–420.
\end{thebibliography}
\end{document}